\documentclass[aps,prb,twocolumn,superscriptaddress]{revtex4}

\usepackage{graphicx}  
\usepackage{subfigure}
\usepackage{multirow}

\usepackage{fancyhdr}
\usepackage{longtable}
\usepackage{lipsum}
\usepackage{parskip}
\usepackage[T1]{fontenc}
\usepackage{dcolumn}   
\usepackage{braket}
\usepackage{bm}        
\usepackage{amsfonts}  
\usepackage{amsmath}   
\usepackage{amssymb}   
\usepackage{natbib}
\usepackage{appendix}
\usepackage{bbm}
\usepackage{hyperref}
\usepackage{booktabs}
\usepackage{comment}
\usepackage{tikz-cd}
\usepackage{booktabs}
\usepackage{adjustbox}
\newtheorem{theorem}{Theorem}
\newtheorem{lemma}{Lemma}
\newtheorem{defn}{Definition}

\newtheorem{col}{Corollary}
\newtheorem{proposition}{Proposition}

\newcommand{\cp}[1]{\mathbb{C}\mathrm{P}^{#1}}

\newcommand{\R}[1]{\mathbb{R}^{#1}}
\newenvironment{proof}{\par\noindent{\it Proof.}}{\hfill$\square$\par\noindent}
\begin{document}


\title{Nash states versus eigenstates for many-body quantum systems}

\author{Chuqiao Lin}
\affiliation{Rudolf Peierls Centre for Theoretical Physics, University of Oxford, Oxford OX1 3PU, United Kingdom}

\author{Vir B. Bulchandani}
\affiliation{Institut f{\"u}r Theoretische Physik, Leibniz Universit{\"a}t Hannover, Appelstra{\ss}e 2, 30167 Hannover, Germany}

\author{S. L. Sondhi}
\affiliation{Rudolf Peierls Centre for Theoretical Physics, University of Oxford, Oxford OX1 3PU, United Kingdom}

\date{\today}
\begin{abstract}
Eigenstates of observables such as the Hamiltonian play a central role in quantum mechanics. Inspired by the pure Nash equilibria that arise in classical game theory, we propose ``Nash states'' of multiple observables as a generalization of eigenstates of single observables. This generalization is mathematically natural for many-body quantum systems, which possess an intrinsic tensor product structure. Every set of observables gives rise to algebraic varieties of Nash state vectors that we call ``Nash varieties''. We present analytical and numerical results on the existence of Nash states and on the geometry of Nash varieties. We relate these ideas to earlier, pioneering work on the Nash equilibria of few-body quantum games and discuss connections to the variational minimization of local Hamiltonians.
\end{abstract}
\maketitle
\tableofcontents
\section{Introduction}

The notion of Nash equilibrium~\cite{nash_equilibrium_1950} provides an appealingly concise formulation of equilibrium under conditions of competition in economics. A simplified mathematical statement of this idea for smooth functions of continuous variables is as follows: we consider $M$ distinct functions $h_i(\vec{x})$ of  $M$ real variables $\vec{x}=(x_1,x_2,\ldots,x_M) \in \mathbb{R}^M$ and require their partial but simultaneous extremization, yielding the $M$ equations
\begin{equation}
\label{eq:locmin}
    {\partial h_i(\vec{x}) \over \partial x_i} = 0,
\end{equation}
whose solution $\vec{x} = \vec{x}_{\rm Nash}$ is a Nash equilibrium.

In economics, the functions $h_i$ might represent the profit earned by firm $i$ as a function of the prices $\vec{x}$ of all the firms in a competitive market including its own price $x_i$. Then the Nash equilibrium conditions Eq. \eqref{eq:locmin} yield a simple model for economically rational decision-making within this market: each firm varies its prices attempting to maximize its profits, all the while assuming that its competitors will keep their prices fixed. In the more abstract language of game theory, $x_i$ defines the ``strategy'' adopted by player $i$ in pursuit of maximizing their ``payoff function'' $h_i$.

Nash's equilibrium is to be contrasted with the standard notion of equilibrium in physics, for example that of energy or free energy, which is instead defined as an extremum of a single function $H(\vec{x})$ of $M$ real variables obtained by solving the $M$ equations 
\begin{equation}
\label{eq:globmin}
    {\partial H(\vec{x}) \over \partial x_i} = 0
\end{equation}
for each $i$. Their solution $\vec{x} = \vec{x}_{\rm Eq}$ yields the equilibrium that is familiar to physicists.

Our choice of notation will suggest to many readers the identification $H(\vec{x}) = \sum_{i=1}^M h_i(\vec{x})$ and while in general there is no reason to make such an identification---we are describing two separate problems---in some cases such an identification can be natural and interesting. In the context of economics, firms might collude in order to maximize their total collective profit. In physics, total energies generally take the form of sums of few-body terms. Nevertheless, optimizing each such term separately in the ``Nash'' sense appears to be a new mathematical idea in that setting, whose quantum formulation will be the main focus of this paper (and is related to the so-called ``joint numerical range''~\cite{szymanski2018classification,xu2024bounding} of the few-body terms).


Thus one way to motivate our work is that it ``quantizes'' the passage from Eq. \eqref{eq:globmin} to Eq. \eqref{eq:locmin}. For Eq. \eqref{eq:globmin}, quantization corresponds to replacing optimization of a function $H$ over a set of real arguments by optimization of the expectation value $\langle \psi | \hat{H} | \psi \rangle$ of an observable $\hat{H}$ over a suitable Hilbert space of pure states $|\psi\rangle$. This is just the familiar variational principle used to find eigenstates and eigenvalues of observables. For the Nash-equilibrium-like conditions Eq. \eqref{eq:locmin}, we similarly propose replacing functions $h_i$ with expectation values of observables $\langle \psi | \hat{h}_i|\psi\rangle$ and replacing optimization over classical variables with optimization over subsets of the quantum degrees of freedom, in a manner that we define precisely in the next section. This leads naturally to a new class of states that we call the ``Nash states'' of a set of observables $\{\hat{h}_i\}$, which generalize the standard notion of eigenstates of a single observable $\hat{H}$. Eigenvalues of the latter similarly generalize to sets of expectation values $\{\langle \psi | \hat{h}_i| \psi \rangle \}$ within a given Nash state $|\psi\rangle$.

This definition raises some immediate questions. First, do Nash states even exist? If so, what are the properties of the set of all Nash states? Second, where in quantum many-body physics do Nash states naturally arise? 

This paper is mostly concerned with answering the first set of questions. We observe that the set of all Nash state vectors corresponding to a given set of operators and allowed quantum strategies forms a real algebraic variety (the ``Nash variety'') whose generic dimension can be estimated by counting constraints. In contrast to eigenstates, which simply yield isolated points in the complex projective space of states, we provide numerical evidence that Nash varieties can generate topologically non-trivial structures in this space. We then exhibit various exceptional cases in which eigenstates of the ``associated Hamiltonian'' $\hat{H} = \sum_{i=1}^M \hat{h}_i$ define Nash states of the set of observables $\{\hat{h}_i\}_{i=1}^M$. While this is certainly not true in general, it is nevertheless true for classes of Hamiltonians for which determining the ground-state energy is a computationally hard task.

On the second question, we make contact with the existing literature on non-cooperative quantum games and show that a subset of Nash states coincides with the game-theoretic Nash equilibria of a ``multi-observable game'' that we define in this paper. Our definition both encompasses the canonical examples of non-cooperative quantum games~\cite{meyer_quantum_1999,eisert_quantum_1999,benjamin_comment_2001,benjamin_multiplayer_2001,kolokoltsov_quantum_2019,johnson_quantum_2002} and greatly expands the set of such games with known Nash equilibria. 
For readers familiar with the earlier literature on non-cooperative quantum games, we emphasize that this progress comes from the viewpoint taken in our work; rather than starting from a shared state and then asking whether its orbit under the space of strategies contains a Nash equilibrium, we effectively ``zoom out'' and consider the space of all possible orbits, which allows us to directly construct the Nash variety for a given set of payoff functions. Some other distinctive features of our formulation compared to previous work on this topic are that it leads naturally to Nash equilibria that are pure states, rather than mixed states, and that it yields ``scalable''~\cite{brassard2005quantum} games that can be extended to arbitrarily many players. We also discuss connections between Nash states and other related ideas that have appeared in the literature in the context of variational quantum algorithms~\cite{Anshu_2021} and so-called ``local minima of Hamiltonians''~\cite{chen2023local}.  

The paper is structured as follows. We first provide a precise definition of both pure and mixed Nash states, and a specialization of this definition that is physically natural for quantum many-body systems. We then discuss the geometry of spaces of Nash states in broad terms before exhibiting an analytically tractable example of a Nash variety that can be visualized in three dimensions and hints at the richness of the general case. We next present some widely studied and physically significant families of Hamiltonians for which eigenstates can be viewed as Nash states, thereby establishing a connection between finding Nash states and solving the local Hamiltonian problem. 

We then turn to the motivating realization of Nash states in the context of non-cooperative quantum games. We explicitly construct Nash varieties for the canonical example of the Quantum Prisoner's Dilemma~\cite{eisert_quantum_1999,benjamin_comment_2001}, yielding geometrical objects that neatly encapsulate an infinite family of pure Nash equilibria for this game. Finally we discuss some applications of approximate, rather than exact, Nash states within the contexts of variational quantum algorithms and the computational complexity theory of quantum games. We conclude with an outline of promising future directions.

\section{From eigenstates to Nash states}
\label{sec:introNashstates}
\subsection{The variational view of eigenstates} \label{sec:var_prob}

Eigenstates of the Hamiltonian $\hat{H}$ in quantum mechanics can be described using the ``variational principle'': eigenstates are stationary points of the function
\begin{equation}
E(|\psi\rangle) = \langle \psi | \hat{H} | \psi \rangle, \quad |\psi \rangle \in \mathcal{H}',
\end{equation}
where the space of states $\mathcal{H}' \cong \cp{d-1}$ comprises unit norm vectors chosen from each ray of the underlying Hilbert space of state vectors $\mathcal{H} \cong \mathbb{C}^d$. We will assume throughout this paper that the Hilbert space dimension $d$ is finite. By definition, ground states of $\hat{H}$ are global minima of $E(|\psi\rangle)$.

There is an alternative and equivalent formulation of the variational principle in terms of Lie group actions on the space of states: fix a reference state $|\psi_0\rangle \in \mathcal{H'}$ and consider the function
\begin{equation}
F(\hat{U}) = \langle \psi_0| \hat{U}^\dagger \hat{H} \hat{U} |\psi_0 \rangle, \quad \hat{U} \in SU(d),
\end{equation}
on $SU(d)$. Then eigenstates $|\psi\rangle = \hat{U}^* |\psi_0\rangle$ of $\hat{H}$ correspond to stationary points $\hat{U}^*$ of $F(\hat{U})$ and ground states to its minima. We can make these elementary observations precise as follows:
\begin{proposition}
    \label{prop:equiv_eigvec_and_variation}
    A state $\ket{\psi}$ is an eigenstate of the Hamiltonian $\hat H$ if and only if
    \begin{equation}
        \label{eq:localstationary}
        \braket{\psi | [\hat H, \hat{A}] | \psi} = 0
    \end{equation}
    for all $\hat{A} \in \mathfrak{su}(d)$. The eigenstate $|\psi\rangle$ is a ground state of $\hat{H}$ if in addition
    \begin{equation}
    \langle \psi | \hat{U}^\dagger \hat{H} \hat{U} | \psi \rangle \geq \langle \psi |\hat{H} | \psi \rangle
    \end{equation}
    for all $\hat{U} \in SU(d)$.
\end{proposition}


Although its proof is standard, this statement of the variational principle in terms of group actions usefully reformulates the problem of finding eigenstates of a given Hamiltonian in terms of finding states whose energy is stationary with respect to a given group of unitary operations. The latter problem is more general than the former; for example one can consider subgroups of the full unitary group $G \leq SU(d)$, or even stationarity of the expectation values of \emph{multiple} observables, as we now discuss.

\subsection{A variational principle for multiple observables}

Let us now consider generalizing the standard variational principle to the case of multiple observables. For this purpose, our formulation of the variational principle in terms of Lie group actions as in Proposition \ref{prop:equiv_eigvec_and_variation} provides a natural starting point. 

Thus suppose that $M$ observables $\hat{h}_i$ act on the same Hilbert space $\mathcal{H} \cong \mathbb{C}^d$ and that for each of these observables, there is a corresponding Lie group $G_i \leq SU(d)$ of unitary operators that also act on $\mathcal{H}$. We will assume that the unitaries corresponding to distinct observables commute, i.e. that $[G_i,G_j]=0$ whenever $i\neq j$.

It will be useful to let $\frak{g}_i$ denote the Lie algebra of $G_i$. Then we define a ``Nash state'' as follows:

\begin{defn}
\label{def:Nash_state}
    A state $|\psi\rangle$ is a \emph{Nash state} of a set of observables $\{\hat{h}_i\}_{i=1}^M$ and corresponding pairwise commuting Lie groups $\{G_i\}_{i=1}^M$ of unitary operators if it satisfies
    \begin{equation}
    \label{eq:Nashstate}
    \langle \psi | [\hat{h}_i,\hat{A}_i] |\psi\rangle = 0, \quad i=1,2,\ldots,M,
    \end{equation}
    for all $\hat{A}_i \in \mathfrak{g}_i$. 
\end{defn}

As discussed in the introduction, we will also be concerned with optimality of Nash states. Depending on the context, it is natural to consider either local or global extrema, and our focus in this paper will largely be minimization rather than maximization of the $\hat{h}_i$. This yields the following definitions:
\begin{defn}
\label{def:loc_Nash_eqstate}
    A Nash state $|\psi\rangle$ of $\{\hat{h}_i\}_{i=1}^M$ and $\{G_i\}_{i=1}^M$ is a \emph{local Nash minimum state}
    if the $M$ bilinear forms
    \begin{equation}
    \label{eq:Nash_locmin}
    B_i(\hat{A}_i,\hat{A}'_i) = \langle \psi | \frac{1}{2}\{\hat{h}_i,\{\hat{A}_i,\hat{A}'_i\}\} -\hat{A}_{i}\hat{h}_i \hat{A}'_i-\hat{A}'_{i}\hat{h}_i \hat{A}_i| \psi\rangle
    \end{equation}
    are positive semidefinite for $\hat{A}_i,\hat{A_i'} \in \mathfrak{g}_i$ and is a \emph{local Nash maximum state} if all these bilinear forms are negative semidefinite.
\end{defn}
\begin{defn}
\label{def:Nash_eqstate}
    A Nash state $|\psi\rangle$ of $\{\hat{h}_i\}_{i=1}^M$ and $\{G_i\}_{i=1}^M$ is a \emph{Nash minimum state}
    if
    \begin{equation}
    \label{eq:Nash_globmin}
    \langle \psi | \hat{U}_i^\dagger \hat{h}_i \hat{U}_i | \psi \rangle \geq  \langle \psi | \hat{h}_i | \psi \rangle, \quad i=1,2,\ldots,M,
    \end{equation}
    for all $\hat{U}_i \in G_i$ and is a \emph{Nash maximum state} if all these inequalities are reversed.
\end{defn}

Letting $\mathcal{M}$ denote the set of density matrices on $\mathcal{H}$, it is clear that Definitions \ref{def:Nash_state}--\ref{def:Nash_eqstate} generalize immediately to density matrices $\hat{\rho} \in \mathcal{M}$ upon replacing the pure-state expectation value $\braket{\psi| [\hat h_i, \hat A_i]| \psi}$ in Eq. \eqref{eq:Nashstate} with a density-matrix expectation value $\operatorname{Tr}(\hat{\rho} [\hat h_i, \hat A_i])$ and similarly for Eqs. \eqref{eq:Nash_locmin} and \eqref{eq:Nash_globmin}. We will primarily be concerned with pure Nash states in this paper but will sometimes refer to Nash density matrices or mixed Nash states in the sense just defined. Note also that in either setting, Definition \ref{def:Nash_eqstate} is stronger than Definition \ref{def:loc_Nash_eqstate}, which in turn is stronger than Definition \ref{def:Nash_state}. 

The idea of simultaneously extremizing the expectation values of a set of multiple observables $\mathcal{S} = \{\hat{h}_i\}_{i=1}^M$ is reminiscent of the ``joint numerical range'' or ``convex support'' of $\mathcal{S}$, which can be defined as the set of tuples of expectation values~\cite{szymanski2018classification,xu2024bounding}
\begin{equation}
J(\mathcal{S}) = \{ (\mathrm{tr}[\hat{\rho}\hat{h}_1], \mathrm{tr}[\hat{\rho}\hat{h}_2],\ldots,\mathrm{tr}[\hat{\rho}\hat{h}_M]) : \hat{\rho} \in \mathcal{M}\} .
\end{equation}
Thus every Nash state in the sense of Definitions \ref{def:Nash_state}-\ref{def:Nash_eqstate} naturally gives rise to a point in the joint numerical range $J(\mathcal{S}) \subset \mathbb{R}^M$. However, for all cases of interest in this paper, the Lie groups $G_i < SU(d)$ are proper subgroups of $SU(d)$ and the corresponding Nash states will not in general correspond to extremal points of the joint numerical range. We will also be concerned with Nash states themselves rather than their images in the joint numerical range (although the latter might be interesting to consider in their own right).

We will frequently specialize to the following case of the above definitions that is directly motivated by the physics of many-body quantum systems. Thus we assume that the observables $\{\hat{h}_i\}_{i=1}^M$ in Definitions \ref{def:Nash_state}--\ref{def:Nash_eqstate} are $k$-local operators~\cite{kitaev2002classical} acting on a system of $N$ qubits. (Note that this includes spatially local operators as a special case, and that the generalization to qudits is straightforward.) We further suppose that the $N$ qubits under consideration are partitioned into $M$ pairwise disjoint blocks of $q$ qubits, that we label with the index $i=1,2,\ldots,M$, such that the support\footnote{Recall that the support $\mathcal{S}$ of an operator $\hat{O}$ on $N$ qubits is the smallest subset of qubits $\mathcal{S} \subset [N]$ such that $\hat{O}$ acts trivially on all the qubits outside $\mathcal{S}$, i.e. the smallest $\mathcal{S}$ such that $\hat{O} = \hat{O}_{\mathcal{S}} \otimes \hat{\mathbbm{1}}_{[N]\backslash \mathcal{S}}$ with $\hat{O}_{\mathcal{S}} \neq \hat{\mathbbm{1}}_{\mathcal{S}}$.} of $\hat{h}_i$ includes block $i$. Thus $N = qM$ for integers $q$ and $M$ with $q \geq k$. Finally, we assign to each $\hat{h}_i$ the Lie group of unitary operations acting strictly on block $i$, so that $G_i \cong SU(2^q)$ and $[G_i,G_j] = 0$ for $i \neq j$. We will refer to this physically motivated special case as the ``local case'' of Definitions \ref{def:Nash_state}--\ref{def:Nash_eqstate}, with a specification of the integers $N$, $1 < M \leq N$, and $k \leq q = \mathcal{O}(N^0)$ always implicit.

We note that when the support of each operator $\hat{h}_i$ coincides exactly with block $i$, the local case of Definition \ref{def:Nash_state} becomes trivial: the corresponding Nash states are simply $M$-fold tensor products of eigenstates of $\hat{h}_i$ restricted to each block. However, if the supports of the operators $\hat{h}_i$ are allowed to exceed block $i$, so that $[\hat{h}_i,\hat{h}_{j}]\neq 0$ for blocks $j$ contiguous with block $i$, the problem becomes non-trivial and even the existence of such Nash states is \textit{a priori} unclear.

Finally, the motivation for our terminology of ``Nash states'' is that every pure (resp. mixed) Nash maximum state in the sense of Definition \ref{def:Nash_eqstate} is in one-to-one correspondence with a pure (resp. mixed) Nash equilibrium in the game theoretic sense. We show this explicitly in Lemma \ref{lemma:NashEq}.

\section{Spaces of Nash states}
\label{sec:geom_nash_variety}
We now explore some general features of spaces of Nash states. Throughout this section, $\mathcal{H} \cong \mathbb{C}^d$ will denote the Hilbert space of pure state vectors, $\mathcal{H}' \cong \cp{d-1}$ the complex projective space of pure states and $\mathcal{M} \subset \mathrm{End}[\mathcal{H}]$ the set of density matrices associated with $\mathcal{H}$. We will follow the convention that the dimension $\mathrm{dim}(A)$ of a real or complex manifold or vector space $A$ always refers to its real dimension. 

\subsection{The convex set of Nash density matrices}
We first consider the set $D \subset \mathcal{M}$ of density matrices that are Nash states. In general, $D$ includes density matrices corresponding to both pure and mixed Nash states. The most basic question concerning $D$ is whether it is non-empty. One can show that $D$ is indeed non-empty either by appealing to topology-based arguments for the existence of mixed Nash equilibria for quantum games~\cite{nash_equilibrium_1950,glicksberg_further_1952,meyer_quantum_1999}, or simply by noting that the maximally mixed state $\hat{\rho}_*= \frac{1}{d} \hat{\mathbbm{1}}$ is always a mixed Nash state. Regarding the global structure of $D$, it follows by the density-matrix version of Definition \ref{def:Nash_state} that $D$ is a convex set, and is therefore always contractible to a point and topologically trivial. 

We can characterize the generic dimension of the set $D$ in more detail as follows. Recall that the bulk real dimension of the space of density matrices is $\mathrm{dim}(\mathcal{M}) = d^2-1$. Meanwhile, for each generator of $G_i$, the density-matrix version of Eq. \eqref{eq:Nashstate} can be viewed as a linear equation in the real and imaginary components of the density matrix. Since there are $\dim(\mathfrak{g}_i)$ such constraints for each $i$, there will be $\sum_{i=1}^M \dim(\mathfrak{g}_i)$ constraints in total. For generic choices of $\hat{h}_i$ and $G_i$, these conditions will be linearly independent and we deduce that the set of mixed Nash states is generically the intersection of $\mathcal{M}$ with $\sum_{i=1}^M \dim(\mathfrak{g}_i)$ linearly independent hyperplanes. Thus 
\begin{equation}
\mathrm{dim}(D) = d^2 - 1 - \sum_{i=1}^M \dim(\mathfrak{g}_i).
\end{equation}

For the local case of Definition \ref{def:Nash_state} on $N$ qubits partitioned into disjoint subsystems of $q=\mathcal{O}(N^0)$ qubits, this becomes
\begin{equation} \label{eq:mixeddimcount}
\mathrm{dim}(D) = 2^{2N} - 1 - (N/q)(2^{2q}-1),
\end{equation}
which is dominated by the exponential in $N$ for large $N$.

\subsection{The algebraic variety of Nash state vectors}

We now turn to the set of pure Nash states, which has a richer structure than the set of Nash density matrices into which it embeds. It will be useful to write $V \subset \mathcal{H}$ for the set of pure Nash state vectors and $V' \subset \mathcal{H}'$ for the complex projectivization of $V$, i.e. the set of pure Nash states. The conditions Eq. \eqref{eq:Nashstate} can be viewed as defining real, homogeneous quadric surfaces in the real and imaginary components of $|\psi\rangle$, implying that $V$ always has the structure of a real algebraic variety. We refer to $V$ as the ``Nash variety''. Its projectivization $V'$ is well-defined by homogeneity of the conditions Eq. \eqref{eq:Nashstate}, but has no natural interpretation as a real or complex projective variety because the constraints Eq. \eqref{eq:Nashstate} are not holomorphic in the state vector $|\psi\rangle$. We would again like to know whether the spaces $V$ and $V'$ are non-empty, and their real dimension at generic points if so. Unfortunately the standard topological arguments for the existence of mixed Nash states do not extend readily to pure Nash states; nevertheless, we exhibit specific examples in Sections \ref{subsec:visualizing} and \ref{sec:solv_spec_cases} for which the existence of pure Nash states can be demonstrated explicitly. 

In terms of global structure, $V$ is generically contractible to a point by homogeneity of Eq. \eqref{eq:Nashstate} and only becomes topologically nontrivial upon complex projectivization to yield $V'$. Regarding the dimension of $V$, the lack of algebraic closure of the real numbers makes it difficult to determine the dimension of a given real algebraic variety systematically~\cite{bochnak2013real}. However, it is still instructive to apply na{\"i}ve dimension counting to the Nash variety, which predicts the following behaviour at generic points. 

For state vectors, the real dimension of Hilbert space $\mathrm{dim}(\mathcal{H}) = 2d$, and the quadratic constraints Eq. \eqref{eq:Nashstate} arising from independent Lie algebra generators will be functionally independent for generic choices of the $\hat{h}_i$. This implies that the Nash variety will have local dimension
\begin{equation}
\label{eq:puredimcount0}
\mathrm{dim}(V) = 2d - \sum_{i=1}^M \dim(\mathfrak{g}_i)
\end{equation}
at generic points, with $\mathrm{dim}(V') = \mathrm{dim}(V)-2$ upon projectivization. For the local case of Definition \ref{def:Nash_state}, this analysis yields the estimate
\begin{equation}
\label{eq:puredimcount}
\mathrm{dim}(V') = 2(2^N-1) - (N/q)(2^q-1),
\end{equation}
which is again exponentially high-dimensional in $N$, as for the set of Nash density matrices.

\subsection{Visualizing the Nash variety for two qubits}
\label{subsec:visualizing}

\begin{figure*}[hbt!]
  \centering
  \subfigure[]{
    \includegraphics[width=0.45\textwidth]{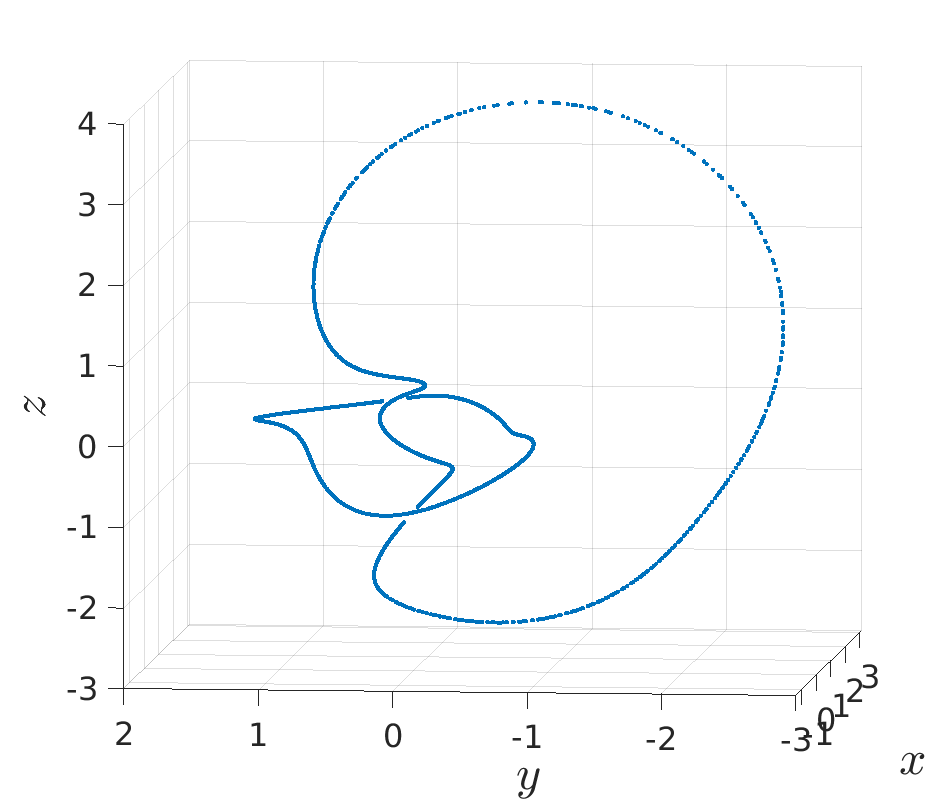}
    \label{fig:linked}
  }
  \subfigure[]{
    \includegraphics[width=0.45\textwidth]{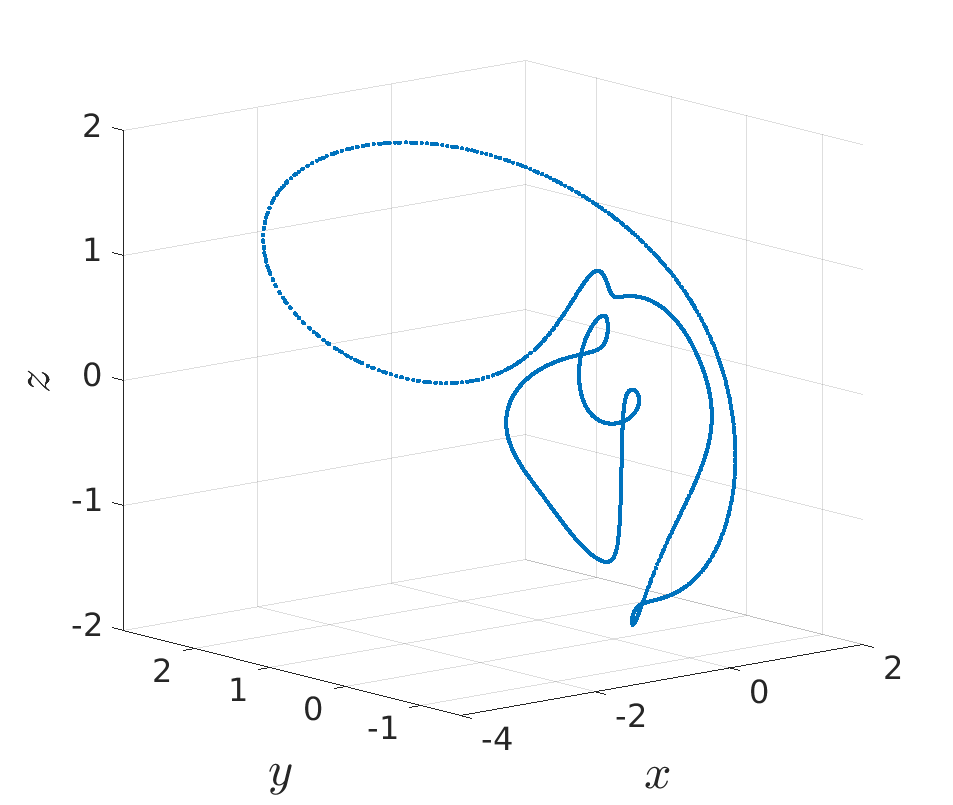}
    \label{fig:notlinked}
  }
  \caption{Illustrative examples of the (double cover of) the set of pure Nash states for $M=2$ real, symmetric observables acting on $N=2$ qubits. We find that this set generically appears to be homeomorphic to two circles and can have non-trivial linking number depending on the realization of the random matrices $\{\hat{h}_1,\hat{h}_2\}$. We obtain this plot by solving for unit normalized $4$-vectors in $S^3$ that lie in the set $\tilde{W}'$ defined in Eq. \ref{eq:defVprime}, which we then stereographically project to $\R{3}$ for ease of visualization.}
  \label{fig:example_variety}
\end{figure*}

We now attempt to validate these generic expectations for the Nash variety by considering low-dimensional examples that are amenable to numerical exploration. To this end, we focus on the local case of Definition \ref{def:Nash_state} consisting of $N$ qubits divided into $N$ single-qubit subsystems, with single-qubit rotations $G_i \cong SU(2)$ acting on each qubit. In this case, Eq. \eqref{eq:puredimcount} reduces to
\begin{equation}
    \label{eq:su2dimcount}
\mathrm{dim}(V') = 2(2^N-1) - 3N.
\end{equation}
Numerical experiments on randomised sets of Hamiltonians are in agreement with Eq. \eqref{eq:su2dimcount}, yielding $\dim (V') = 0$ for $N=2$ and $\dim(V') = 5$ for $N=3$ qubits. In more detail, we use the homotopy continuation method of numerical algebraic geometry (specifically the package \texttt{HomotopyContinuation.jl}\cite{breiding_homotopycontinuationjl_2018}\cite{sampling_bottlenecks2023}) for numerically finding and projectivizing these algebraic varieties, as follows. 
In our analysis, Eq. \eqref{eq:Nashstate} is first interpreted as imposing constraints on the $2d$ real degrees of freedom of the state vector $\ket{\psi}\in \mathbb{C}^d$. The projectivization of $\mathbb{C}^d$ to $\cp{d-1}$ is executed by normalizing and manually setting one real coordinate to zero (equivalent to specifying a local coordinate chart). The raw output of the algorithm\cite{breiding_homotopycontinuationjl_2018} yields real solutions amid complex ones; the latter are unphysical and we correspondingly remove them from the dataset by post-selection. For $N=M=2$, generic samples yield a discrete set of solutions, consistent with dimension $0$, while for $N=M=3$, a discrete set of solutions only appears when intersecting the Nash variety with a linear subspace of codimension $5$. Linearizing about these known solutions then generates a five-dimensional tangent space, from which a five-dimensional coordinate patch locally covering $V'$ can be constructed iteratively by a predictor-corrector method, i.e. moving linearly in the tangent space and subsequently falling back into the manifold using Newton's method. In general, this scheme allows the geometry of $V'$ to be probed locally, but the computation time increases rapidly with $N$ and $M$ and global information about $V'$ is correspondingly difficult to obtain.

Such analytical and numerical explorations suggest that the lowest-dimensional nontrivial $V'$ ($N=2$ qubits) consists of isolated points, while the next lowest dimensional example ($N=3$ qubits) is too high-dimensional to visualize easily. We can circumvent these difficulties and find a two-qubit example for which the set of pure Nash states $V'$ is both nontrivial and possible to visualize by imposing time-reversal symmetry on the observables $\hat{h}_i$, i.e. we demand that the operators $\hat{h}_i$ are real and symmetric as matrices. 

Let us first parameterize state vectors $|\psi\rangle \in \mathbb{C}^4$ by their real and imaginary parts, i.e.
\begin{equation}
    \ket{\psi} = \mathbf{x} + i\mathbf{y}, \quad (\mathbf{x}, \mathbf{y}) \in \mathbb{R}^8.
\end{equation}
For real and symmetric $\hat{h}_i$, the operators $\hat{B}_{iX} = [\hat h_i, \hat{X}_{i}], \,B_{iZ} = [\hat h_i, \hat Z_{i}]$ are real and antisymmetric while the operators $B_{iY} = \mathrm{i}[\hat h_i, \hat Y_i]$ are real and symmetric. This yields a particularly simple and explicit form for the conditions Eq. \eqref{eq:Nashstate}, which reduce to solving six simultaneous quadratic equations over $\mathbb{R}^8$:

\begin{align}
    \mathbf{x}^T B_{iY} \mathbf{x} + \mathbf{y}^T B_{iY} \mathbf{y} &= 0, \quad i=1,2, \label{eq:iy}\\
    \mathbf{x}^T B_{iX} \mathbf{y} &= 0, \quad i=1,2, \label{eq:ix}\\
    \mathbf{x}^T B_{iZ} \mathbf{y} &= 0,\quad i=1,2. \label{eq:iz}
\end{align}

This set of equations defines a real algebraic variety $V \subseteq \mathbb{R}^8$. In addition to the two-dimensional components expected from the generic dimension counting argument Eq. \eqref{eq:puredimcount0}, we argue that $V$ also contains components $\tilde{V} \subset V$ with dimension 3, thanks to time-reversal symmetry of $\hat{h}_i$. We do this constructively, by letting

\begin{equation}
\label{eq:tildeV}
\tilde{V} = \{(\mathbf{x},\lambda \mathbf{x}) \in \mathbb{R}^8 : \mathbf{x} \in \mathbb{R}^4, \lambda \in \mathbb{R}, \, \mathbf{x}^T B_{iY} \mathbf{x}=0\}.
\end{equation}

One can verify directly that $\tilde{V} \subset V$. Indeed, once we make the assumption that $\mathbf{y}$ and $\mathbf{x}$ are parallel as in Eq. \eqref{eq:tildeV}, the constraints in Eqs. \eqref{eq:ix} and \eqref{eq:iz} immediately become redundant. Thus, for generic choices of $\hat{h}_i$, $\tilde{V}$ is the product of the intersection of two independent quadric surfaces in $\mathbb{R}^4$ with the real line and has local dimension three. If we instead allow for linearly independent $\mathbf{x}$ and $\mathbf{y}$, we find that numerically solving Eqs. \eqref{eq:iy}-\eqref{eq:iz} in $\R{8}$ yields two-dimensional components distinct from $\tilde{V}$. We therefore conjecture that generic points of $V$ lie in the three-dimensional component $\tilde{V}$.

This conjecture implies that for generic choices of $\hat{h}_i$, the complex projectivization $V'$ locally has the structure of a 1-manifold. In more detail, it implies that generic points of $V'$ lie in the complex projectivization of $\tilde{V}'$, which is the quotient of $\tilde{V}$ by the equivalence relation
\begin{equation}
\label{eq:cproj}
(\mathbf{x},\mathbf{y}) \sim r (\cos{\theta}\mathbf{x}-\sin{\theta} \mathbf{y}, \sin{\theta} \mathbf{x} + \cos{\theta} \mathbf{y})\end{equation}
for $r>0$ and $\theta \in [0,2\pi)$.
Under this relation, 
\begin{equation}
\label{eq:quot}
(\mathbf{x},\lambda \mathbf{x}) \sim (\mathbf{x},0) \sim (-\mathbf{x},0)
\end{equation}
for all $\lambda \in \mathbb{R}$, from which it follows that a two-to-one parameterization of $\tilde{V}'$ is given by
\begin{equation}
\label{eq:defVprime}
\tilde{W}' = \{(\mathbf{x},0) \in \mathbb{R}^8 : \mathbf{x} \in S^3, \, \mathbf{x}^T B_{iY} \mathbf{x}= 0\},
\end{equation}
where $S^3 = \{\mathbf{x} \in \mathbb{R}^4 : \| \mathbf{x} \| = 1\}$ denotes the unit $3$-sphere in $\mathbb{R}^4$. Identifying the points of $\tilde{W}'$ under the remaining freedom allowed by the equivalence relation Eq. \eqref{eq:quot} finally yields $\tilde{V}' = \tilde{W}'/\{\mathbf{x} \sim -\mathbf{x}\}$. We deduce that $\tilde{V}'$ is locally a 1-manifold.

It is easier to visualize the double cover $\tilde{W}'$ (viewed as a submanifold of $S^3$) than it is to visualize $\tilde{V}'$, and we do this via stereographic projection onto $\R{3}$. Numerical evidence suggests that $\tilde{W}'$ can be topologically non-trivial, as follows. Points of $\tilde{W}'$ are obtained via homotopy continuation. We find that for random choices of the observables $\hat{h}_i$, the Nash variety appears to be homeomorphic to a pair of circles. Whether or not these circles are linked depends on the choice of $\hat{h}_i$. Figure \ref{fig:example_variety} depicts examples of both scenarios for distinct, randomly chosen pairs of real, symmetric $N=2$ qubit observables $\{\hat h_1, \hat h_2\}$. 

\section{Frustration-free and strictly $k$-local Hamiltonians} \label{sec:solv_spec_cases}
We now discuss certain special cases in which Nash states in the sense of Definition \ref{def:Nash_state} reduce either to conventional eigenstates or to local extrema of the associated Hamiltonian $\hat{H}=\sum_{i=1}^M \hat{h}_i$, in the sense of Ref. \onlinecite{Anshu_2021} or Ref. \onlinecite{chen2023local}. Throughout this section we will restrict our attention to $k$-local Hamiltonians on $N$ qubits with $k = \mathcal{O}(N^0)$.

We first consider so-called ``frustration-free Hamiltonians''~\cite{hastings2006solving,perezgarcia2007matrix,Sattath_2016}. Recall that a Hamiltonian $\hat{H}$ is frustration-free if it has the following property~\cite{tasaki2020physics}.
\begin{defn}
\label{def:FF}
A Hamiltonian $\hat{H}$ is \emph{frustration-free} if there exist observables $\{\hat{h}_i\}_{i=1}^M$ with corresponding minimum eigenvalues $\{\epsilon_i\}_{i=1}^M$ such that
\begin{enumerate}
    \item $\hat{H}= \sum_{i=1}^M \hat{h}_i$,
    \item the operators $\hat{h}_i - \epsilon_i\hat{\mathbbm{1}}$ are positive semi-definite for all $i$,
    \item there exists a state $|\psi\rangle$ for which
\begin{equation}
\label{eq:GS}
\hat{h}_i |\psi\rangle = \epsilon_i |\psi\rangle
\end{equation}
for all $i$.
\end{enumerate}
\end{defn}
Families of frustration-free Hamiltonians of interest to condensed matter physicists are usually extensive, so that $M = \Omega(N)$ in Definition \ref{def:FF}. The simplest case of frustration-freeness is commuting terms, $[\hat{h}_i,\hat{h}_j] = 0$ for all $i,\,j$. Definition \ref{def:FF} implies that if $\hat{H}$ is frustration-free, a state satisfying Eq. \eqref{eq:GS} for all $i$ is a ground state of $\hat{H}$, and vice-versa. The following result is immediate upon applying Proposition \ref{prop:equiv_eigvec_and_variation} to each $\hat{h}_i$ in turn.
\begin{proposition}
\label{prop:FF}
Let $\hat{H}$ be frustration-free with the operators $\{\hat{h}_i\}_{i=1}^M$ as in Definition \ref{def:FF}. Then ground states of 
$\hat{H}$ are Nash minimal states of $\{\hat{h}_i\}_{i=1}^M$ with respect to any set of pairwise commuting Lie groups $\{G_i\}_{i=1}^M$.
\end{proposition}

We next consider $k$-local Hamiltonians in the sense of quantum complexity theory~\cite{kitaev2002classical}, a class that is broad enough to encompass most of the lattice models studied in quantum condensed matter physics. As usual~\cite{kitaev2002classical}, we define a Hamiltonian $\hat{H} = \sum_{i=1}^{M'} \hat{h}'_i$ on $N$ qubits to be \emph{k-local} if each $\hat{h}'_i$ acts non-trivially on at most $k$ qubits. We define $\hat{H}$ to be \emph{strictly k-local} if each $\hat{h}'_i$ acts non-trivially on exactly $k$ qubits. This is a more restricted class of Hamiltonians than $k$-local Hamiltonians but nevertheless of physical interest; a textbook example is the $SU(2)$-symmetric Heisenberg model~\cite{bethe1931theorie} which is strictly $k$-local with $k=2$. We then have the following result:
\begin{theorem}
\label{thm:eigenstatethm}
Let $\hat{H}$ be a strictly $k$-local Hamiltonian on $N$ qubits. Let $G_i \cong SU(2)$ denote the group of single-qubit rotations acting on qubit $i$. Then there exist observables $\{\hat{h}_i\}_{i=1}^N$ whose associated Hamiltonian is $\hat{H} = \sum_{i=1}^N \hat{h}_i$, such that
\begin{enumerate}
    \item All eigenstates of 
    $\hat{H}$ are Nash states of $\{\hat{h}_i\}_{i=1}^N$ and $\{G_i\}_{i=1}^N$.
    \item All ground states of 
    $\hat{H}$ are Nash minimum states of $\{\hat{h}_i\}_{i=1}^N$ and $\{G_i\}_{i=1}^N$.
\end{enumerate}
\end{theorem}
\begin{proof}
For simplicity we consider strictly $2$-local Hamiltonians; the generalization to strictly $k$-local Hamiltonians with $k \neq 2$ is straightforward. Thus let
\begin{equation}
\hat{H} = \sum_{\{i,j\}\in E} \hat{h}_{ij}
\end{equation}
be a strictly 2-local Hamiltonian on an undirected graph $(V,E)$ with vertices or sites $V$ and edges $E$, where the total number of qubits $N = |V|$.

To prove the first statement, let $|\psi\rangle$ be an eigenstate of $\hat{H}$ such that $\hat{H}|\psi\rangle = E_{\psi} |\psi\rangle$. Then for all linear operators $\hat{A}$ we have
\begin{align}
\label{eq:stat}
\langle \psi| [\hat{H},\hat{A}]|\psi\rangle = 0.
\end{align}
Now consider acting at site $i$ with an arbitrary single-qubit rotation $\hat{U}_i(\mathbf{v}_i) = e^{\hat{A}}$, where $\hat{A} = \mathrm{i}\mathbf{v}_i \cdot \hat{\pmb{\sigma}}_i$ and $\sigma_i^\alpha$ denotes the Pauli matrices acting on qubit $i$. Then Eq. \eqref{eq:stat} implies that
\begin{equation}
\label{eq:locstat}
\langle \psi | [\hat{H},\hat{\sigma}_i^\alpha] |\psi \rangle = 2\langle \psi | [\hat{h}_{i},\hat{\sigma}_i^\alpha] |\psi \rangle = 0, \quad \alpha=1,2,3,
\end{equation}
at each site $i$, where the ``star Hamiltonians'' 
\begin{equation}
\label{eq:locH}
\hat{h}_i = \frac{1}{2}\sum_{\{i,j\} \in E} \hat{h}_{ij}, \quad i = 1,2,\ldots,N,
\end{equation}
satisfy $\sum_{i=1}^N \hat{h}_i = \hat{H}$. Thus $|\psi\rangle$ is a Nash state of $\hat{H}$ with respect to the observables $\{\hat{h}_i\}_{i=1}^N$ and single-qubit unitaries acting on each qubit. 

To prove the second statement, let $|\psi\rangle$ be a ground state of $\hat{H}$ and note that by definition, arbitrary single-qubit unitaries $\hat{U}_i$ acting at site $i$ satisfy
\begin{equation}
\langle \psi | \hat{U}_i^\dagger \hat{H} \hat{U}_i |\psi \rangle \geq \langle \psi | \hat{H} |\psi \rangle.
\end{equation}
This implies that
\begin{equation}
\sum_{\{i,j\} \in E} \langle \psi | \hat{U}_i^\dagger \hat{h}_{ij} \hat{U}_i | \psi \rangle \geq \sum_{\{i,j\} \in E} \langle \psi |  \hat{h}_{ij} | \psi \rangle, 
\end{equation}
and the result follows by definition of $\hat{h}_i$.
\end{proof}
Note that Theorem \ref{thm:eigenstatethm} relies crucially on our specific choice of $\hat{h}_i$ in Eq. \eqref{eq:locH}, specifically the equal weight assigned to each edge of the underlying graph. If one relaxes the condition of strict $k$-locality, such a choice may no longer be possible. For example, consider Hamiltonians with both 2- and 1- local contributions,
\begin{equation}
\hat{H} = \sum_{\{i,j\}\in E} \hat{h}_{ij} + \sum_{i \in V} \hat{s}_i.
\end{equation}
In this case, the steps leading to the proof of Theorem \ref{thm:eigenstatethm} imply that $|\psi\rangle$ is a Nash ground state with respect to the star Hamiltonians
\begin{equation}
\hat{h}_i = \frac{1}{2}\sum_{\{i,j\}\in E} \hat{h}_{ij} + \frac{1}{2}\hat{s}_i.
\end{equation}
However, for generic choices of $\hat{s}_i$, there is no choice of coefficients $\alpha_i$ such that
\begin{equation}
\hat{H} = \sum_{i=1}^N \alpha_i \hat{h}_i,
\end{equation}
and these steps do not yield an analogue of Theorem \ref{thm:eigenstatethm}.

Finally, we note that in the special case that $|\psi\rangle$ is a product state and $\hat{H}$ is strictly $k$-local, the reasoning leading to Theorem \ref{thm:eigenstatethm} implies that $|\psi\rangle$ being a Nash state is equivalent to $|\psi\rangle$ being ``locally optimal'' in the sense of Ref. \onlinecite{Anshu_2021}, which is turn a special case of being a ``local minimum under local unitary perturbations'' in the sense of Ref. \onlinecite{chen2023local}. The following stronger statement also holds.
\begin{col}
\label{cor:prod}
Let $\hat{H}$ be strictly $k$-local and let $\{\hat{h}_i\}_{i=1}^N$ and $\{G_i\}_{i=1}^N$ be as in Theorem \ref{thm:eigenstatethm}. Then optimal product-state approximations to the ground state energy of $\hat{H}$, i.e. product states $|\psi^* \rangle = \otimes_{i=1}^N |\phi^*_i\rangle$ such that
\begin{equation}
\langle \psi^* | \hat{H} | \psi^* \rangle = \min_{|\psi \rangle = \otimes_{i=1}^N |\phi_i\rangle} \langle \psi | \hat{H} | \psi \rangle
\end{equation}
are Nash minimum states of $\{\hat{h}_i\}_{i=1}^N$ and $\{G_i\}_{i=1}^N$.
\end{col}

A striking conclusion of these combined results is that strictly $k$-local Hamiltonians $\hat{H}$ give rise to (at least) two distinct Nash minimum states: the ground state of $\hat{H}$ and the optimal product-state approximation to the ground state of $\hat{H}$. (Recall that the former and the latter generically differ in energy except in the limits of very high degree~\cite{brandao2013product} or very large qudit dimension~\cite{lieb1973classical,bulchandani2024classical} and whenever they do differ, the energy of $\hat{H}$ yields an entanglement witness at sufficiently low temperatures\cite{Dowling_2004}.) Finally, we expect that such ``solvable'' examples of Nash states generically define merely an exponentially large family of discrete points in the uncountably large algebraic variety of Nash states that was discussed in the previous section.

Together, these results imply that determining all the points in a given $N$-qubit Nash variety is at least as hard as solving the $N$-qubit ``local Hamiltonian problem'' from quantum computational complexity theory~\cite{kitaev2002classical,piddock2015complexity}. We will revisit this point below in the context of variational quantum algorithms for Hamiltonians within the purview of Theorem \ref{thm:eigenstatethm}.

\section{Nash states at non-zero temperature}
\label{sec:temperature}
In the previous section, we exhibited some specific families of local Hamiltonians $\hat{H} = \sum_{i=1}^M \hat{h}_i$ for which ground states could be viewed as Nash minimum states of $\{\hat{h}_i\}_{i=1}^M$. A natural question in the context of physical realizations of such ground states is how far this ``minimality'' property extends to small but non-zero temperatures. In this setting, global minimality in the sense of Definition \ref{def:Nash_eqstate} is an overly strong requirement. We instead ask the following question: does the property of being a local Nash minimum state, in the sense of Definition \ref{def:loc_Nash_eqstate}, persist to non-zero temperatures?

For concreteness, consider an $N$-qubit state $|\psi\rangle$ that is a Nash minimum state with respect to the local operators $\{\hat{h}_i\}_{i=1}^N$ and Lie groups of unitary operators $G_i$ and simultaneously a ground state of their associated Hamiltonian $\hat{H} = \sum_{i=1}^N \hat{h}_i$. Then picking a basis $\{\hat{A}_{ia}\}_{a=1}^{\dim(\mathfrak{g}_i)}$ for each Lie algebra $\mathfrak{g}_i$, a sufficient condition for the local Nash minimum state property of Definition \ref{def:loc_Nash_eqstate} to persist to low-temperature thermal states $\hat{\rho}_{\beta}= \frac{1}{Z_{\beta}}e^{-\beta \hat{H}}$ with $Z_\beta = \mathrm{tr}[e^{-\beta \hat{H}}]$ is that the Hessian 
\begin{equation} \label{eq:Hessian_condition}
    \mathrm{HS}^{(i)}_{ab} = \mathrm{tr}\left[\hat{\rho}_\beta\left( \frac{1}{2}\{\hat{h}_i,\{\hat{A}_{ia}, \hat{A}_{ib}\}\} - \hat{A}_{ia} \hat  h_i  \hat{A}_{ib}-\hat{A}_{ib} \hat  h_i  \hat{A}_{ia} \right)\right]
\end{equation}
is positive semidefinite for all $i = 1, 2, ..., N$ and sufficiently low temperatures.

Specifically, we consider the one-dimensional transverse field Ising model (TFIM) with Hamiltonian
\begin{equation}
    \hat H_{\mathrm{TFIM}} = -\sum_{i=1}^N \hat Z_i \hat Z_{i+1} - g\sum_{i=1}^N \hat X_i
\end{equation}
and periodic boundary conditions. We can write this as $\hat H_{\mathrm{TFIM}} = \sum_{i=1}^N \hat h_i$, where
\begin{equation}
    \label{eq:local_decomp}
    \hat h_i =  - \frac{1}{2}\hat Z_i (\hat Z_{i-1} + \hat Z_{i+1}) - g\hat  X_i.
\end{equation}
One can show that the ground state and eigenstates of $\hat{H}$ define Nash minimum states and Nash states of $\{\hat{h}_i\}_{i=1}^N$ with respect to single-qubit rotations $G_i \cong SU(2)$ on qubit $i$, see Appendix \ref{app:GSareNE}. 

The 3-by-3 Hessians $\mathrm{HS}_i$ are identical for all $i$ by translation invariance and can be calculated analytically,
\begin{equation} \label{eq:isinghessian}
    \mathrm{HS}^{(i)} = \begin{pmatrix}
        4 \langle z z \rangle_\beta & 0 & 0 \\
        0 & 4\langle z z \rangle_\beta +4g \langle x\rangle_\beta & 0 \\
        0 & 0 & 4g\langle x\rangle_\beta
    \end{pmatrix}
\end{equation}
where we have used the shorthands $\langle z z\rangle_\beta  \equiv \braket{\hat Z_i \hat Z_{i+1}}_\beta$ and $\langle x\rangle_\beta \equiv \braket{\hat X_i}_\beta$ for thermal averages. At zero temperature this matrix must be positive definite by the Nash minimum state property of the TFIM ground state, and by non-negativity of the expectation values $\langle zz \rangle_{\beta}$ and $\langle x\rangle_{\beta}$ for $g>0$ and any finite temperature, it follows by Eq. \eqref{eq:isinghessian} that Gibbs states are local Nash minima for any $T>0$. See Fig. \ref{fig:XZcor} for an illustration. Another way to interpret this result is that if we sample a ``representative state'' by filling Bogoliubon orbitals according to Fermi-Dirac statistics, then this eigenstate is almost surely a local Nash minimum in the thermodynamic limit\cite{khemani_eigenstate_2014}. For completeness, details of the calculation of $\langle x\rangle_\beta$ and $\langle zz\rangle_\beta$ in Eq, \eqref{eq:isinghessian} for the TFIM are summarized in Appendix \ref{app:hess}.

\begin{figure}
    \centering
    \includegraphics[width=0.5\textwidth]{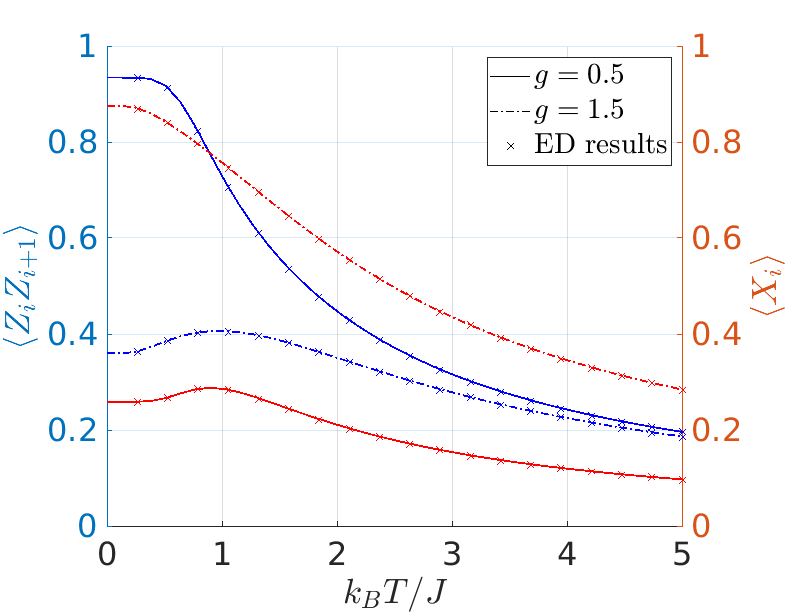}
    \caption{The finite temperature correlators $\braket{zz}_\beta$ and $\braket{x}_\beta$ plotted against temperature $k_BT/J$ with $J=1$. The solid/dashed lines use Eq. \eqref{eq:thermalaveraged2pt} to evaluate the thermal expectation values at $g=0.5, 1.5$ respectively, with $N=10$. The results are consistent with exact diagonalization.}
    \label{fig:XZcor}
\end{figure}

\section{Connection to quantum games} \label{sec:connection}
\subsection{Nash maximum states are Nash equilibria}
We now show as promised that every single Nash maximum state in the sense of Definition \ref{def:Nash_state} corresponds to a pure-strategy Nash equilibrium of a multi-player quantum game, that we call the ``multi-observable game''. Nash equilibria of non-cooperative quantum games were introduced in Refs. \onlinecite{meyer_quantum_1999} and \onlinecite{eisert_quantum_1999}. Such games are distinct from so-called quantum nonlocal games, which are usually cooperative in nature~\cite{brassard2005quantum}.

We define the multi-observable game as follows.
\begin{defn}
\label{def:MOG}
The \emph{$(M,N)$-multi-observable game} is an $M$-player quantum game, in which players have knowledge of a set of $M$ observables $\hat{h}_i$ and access to a shared pure state $|\psi_0\rangle$ on $N$ qubits. During their turn, Player $i$ is allowed to act on the shared state with unitary operators from a compact Lie group $G_i$, where $[G_i,G_j] = 0$ for $j \neq i$. Let $|\psi\rangle$ denote the shared pure state after all the players have completed their turns and
\begin{equation}
u_i = \langle \psi | \hat{h}_i | \psi \rangle.
\end{equation}
Then player $i$ seeks to maximize $u_i$ and wins the game if $u_i \geq u_j$ for all $j \neq i$.
\end{defn}
We will refer to any full specification of the data $\mathcal{G} = (M,\,N,\,|\psi_0\rangle,\{G_i\}_{i=1}^M,\{\hat{h}_i\}_{i=1}^M)$ in Definition \ref{def:MOG} as a ``multi-observable game instance''. More formally, every multi-observable game instance defines a ``non-cooperative static game with complete information''~\cite{gibbons1992game,lee_efficiency_2003,bostanci_quantum_2022}, whose normal form is the tuple $(\{G_i\}_{i=1}^M; \{u_i\}_{i=1}^M)$, where the strategy space for player $i$ is the Lie group $G_i$ and the payoff function $u_i: \prod_{j=1}^M G_j \to \R{}$ for player $i$ is given by the expectation values of the observable $\hat{h}_i$, namely
\begin{equation}
u_i(\hat{U}_1,\hat{U}_2,\ldots,\hat{U}_M) = \langle \psi_0 | \left(\prod_{j=1}^M \hat{U}_j^\dagger\right)\hat h_i \left(\prod_{j=1}^M \hat{U}_j\right)|\psi_0\rangle.
\end{equation}
As usual~\cite{lee_efficiency_2003,bostanci_quantum_2022}, the multi-observable game can be formulated as a purely classical game, with an inevitable exponential-in-$N$ computational overhead that arises from encoding the payoff functions classically. All standard notions of classical game theory, such as Nash equilibria, thus apply to the multi-observable game.

Note that our requirement that the Lie groups $G_i$ are pairwise commuting means that there is no difference between the ``static'' and the ``dynamic'' formulations~\cite{gibbons1992game} of the multi-observable game: the order in which the players play the game is immaterial. We have correspondingly formulated Definition \ref{def:MOG} as a static game. If this requirement of pairwise commutativity of the $G_i$ is relaxed, then the multi-observable game becomes a dynamic game and Definition \ref{def:MOG} must be modified accordingly.

Having defined the multi-observable game, we first verify our earlier claim that the Nash equilibrium conditions for this game recover the Nash maximum states of Definition \ref{def:Nash_eqstate}. A formal statement of this equivalence is as follows.
\begin{lemma}
\label{lemma:NashEq}
Pure-strategy Nash equilibria of the multi-observable game Definition \ref{def:MOG} are in one-to-one correspondence with Nash maximum states in the sense of Definition \ref{def:Nash_eqstate}.
\end{lemma}
\begin{proof}
For the forward implication, note that a pure-strategy Nash equilibrium of a multi-observable game instance $\mathcal{G} = (M,\,N,\,|\psi_0\rangle,\{G_i\}_{i=1}^M,\{\hat{h}_i\}_{i=1}^M)$ consists~\cite{gibbons1992game} of a strategy profile $(\hat{V}_1,\hat{V}_2,\ldots,\hat{V}_M) \in \prod_{j=1}^M G_j$ that satisfies
\begin{align}
\label{eq:stratNE}
\langle \psi_0 | \hat{V}_{-i}^\dagger \hat{U}_i^\dagger  \hat{h}_i \hat{U}_i\hat{V}_{-i} |\psi_0 \rangle
\leq \langle \psi_0 | \hat{V}_{-i}^\dagger \hat{V}_i^\dagger  \hat{h}_i \hat{V}_i\hat{V}_{-i} |\psi_0 \rangle
\end{align}
for all $i=1,2,\ldots,N$ and all $\hat{U}_i \in G_i$, where $
\hat{V}_{-i} = \prod_{j\neq i}\hat{V}_j$. It follows that the state $|\psi\rangle = \prod_{i=1}^M \hat{V}_i|\psi_0\rangle$ is a Nash maximum state of the operators $\{\hat{h}_i\}_{i=1}^M$ and the Lie groups $\{G_i\}_{i=1}^M$ in the sense of Definition \ref{def:Nash_state}. For the reverse implication, let $|\psi_0\rangle$ be a Nash maximum state of the operators $\{\hat{h}_i\}_{i=1}^M$ with respect to the Lie groups $\{G_i\}_{i=1}^M$ and note that the ``do nothing'' strategy profile $(\hat{\mathbbm{1}},\hat{\mathbbm{1}},\ldots,\hat{\mathbbm{1}})$, in which all players act with the identity, defines a Nash equilibrium of the multi-observable game instance $\mathcal{G} = (M,\,N,\,|\psi_0\rangle,\{G_i\}_{i=1}^M,\{\hat{h}_i\}_{i=1}^M)$.
\end{proof}

We will have this equivalence in mind whenever we refer to Nash maximum states of a given multi-observable game instance $\mathcal{G}$. As in Section \ref{sec:introNashstates}, Definition \ref{def:MOG} and Lemma \ref{lemma:NashEq} extend straightforwardly from pure states to density matrices.

\subsection{A case study: the Quantum Prisoner's Dilemma}
\label{sec:case_study}
\begin{table}[t]
    \centering
    \begin{tabular}{|c|cc|}
    \hline
 & \textsc{Cooperate} & \textsc{Defect}\\
 \hline
         \textsc{Cooperate} &   $(3,3)$&$(0,5)$\\
         \textsc{Defect} & 
     $(5,0)$&$(1,1)$ \\
     \hline
     \end{tabular}
\caption{Payoff matrix for the classical Prisoner's Dilemma}
\label{tab:prisoner_dilemma}
\end{table}

We now illustrate the utility of Definition \ref{def:MOG} by showing how one of the canonical examples of non-cooperative quantum games~\cite{eisert_quantum_1999,benjamin_comment_2001,meyer_quantum_1999,kolokoltsov_quantum_2019}, the \textit{Quantum Prisoner's Dilemma}, can be realized and analyzed as an instance of the multi-observable game. We choose the same payoff matrix to set up the comparison with the classical Prisoner's Dilemma as were considered in Ref. \onlinecite{eisert_quantum_1999}, see Table \ref{tab:prisoner_dilemma}. Classically, this game admits the dominant strategy $(\textsc{Defect}, \textsc{Defect})$, which is also the Nash equilibrium. The Prisoner's Dilemma can be quantized by representing each player's strategy by a qubit, with the ket $\ket{0}$ corresponding to the strategy $\textsc{Cooperate}$ and the ket $\ket{1}$ corresponding to the strategy $\textsc{Defect}$. Starting from a (possibly entangled) two-qubit state, corresponding to a superposition of possible strategies for the players, each player is allowed to apply single-qubit unitaries to their qubit. Payoffs are then assigned according to the outcome of a measurement in the computational basis and the corresponding entry of Table \ref{tab:prisoner_dilemma}. 

This defines a family of $(2,2)$-multi-observable games $\mathcal{G}_{\mathrm{QPD}}$ in the sense of Definition \ref{def:MOG}, where the Lie groups $G_i \cong SU(2)$ correspond to single-qubit rotations of each player's qubit and the observables defining each player's payoff function are given by
\begin{equation}
\label{eq:qpdpayoffs}
    \hat{h}_1 = \begin{pmatrix}
        3 & 0 & 0 & 0 \\
        0 & 0 & 0 & 0 \\
        0 & 0 & 5 & 0 \\
        0 & 0 & 0 & 1
    \end{pmatrix}, \quad
    \hat{h}_2 = \begin{pmatrix}
        3 & 0 & 0 & 0 \\
        0 & 5 & 0 & 0 \\
        0 & 0 & 0 & 0 \\
        0 & 0 & 0 & 1
    \end{pmatrix}.
\end{equation}

In the following, we will solve for the Nash variety and the set of Nash maximum states of $\{\hat{h}_1,\hat{h}_2\}$ with respect to single-qubit rotations. As in Section \ref{subsec:visualizing},  we first argue that it suffices to consider ``rebits'', i.e. state vectors
\begin{equation}
    \label{eq:rebits}
    \ket{\psi} =X_{0}\ket{00} +X_{1}\ket{01} +X_{2}\ket{10} + X_{3} \ket{11} 
\end{equation}
with $(X_{0}, X_{1}, X_{2}, X_{3}) \in \mathbb{R}^4$. This restriction is justified carefully in Appendix \ref{app:rebits}. In these coordinates, the Nash state conditions Eq. \eqref{eq:defVprime} define the real algebraic variety
\begin{equation}
    \label{eq:NV_qpd}
    \left\{
    \begin{aligned}
        & 2 X_0 X_2 + X_1 X_3 = 0, \\
        & 2X_0 X_1 + X_2 X_3 = 0.
    \end{aligned}
    \right.
\end{equation}

For the sake of visualization and by homogeneity of Eq. \eqref{eq:defVprime}, we now restrict our attention to unit normalized rebits $X_0^2+X_1^2+X_2^2 + X_3^2=1$ and use stereographically projected coordinates $(x,y,z) \in \mathbb{R}^3$ to represent Nash states, with
\begin{equation}
    \label{eq:stereoprojection}
    x_i = \frac{X_{i}}{1+X_{0}},
\end{equation}
the usual identifications $x_1\equiv x, \, x_2 \equiv y, \, x_3 \equiv z$, and the inverse transformation given by
\begin{equation}
    \label{eq:invstereoprojection}
    X_{i} = \frac{2x_{i}}{1+r^{2}}, \quad X_{0} = \frac{1-r^{2}}{1+r^{2}},
\end{equation}
where $r^2=x^2+y^2+z^2$. Making this substitution in Eq. \eqref{eq:NV_qpd} and plotting the resulting solutions, the (unit normalized) Nash variety exhibits four one-dimensional components that are intersecting or linked, see Figure \ref{fig:nvqpd} for an illustration. For a Nash state $\ket{\psi}$ to also be a Nash maximum state, it must further satisfy the following inequality constraints:
\begin{equation}
    \label{eq:Nashmax4qpd}
    \left\{
    \begin{aligned}
        &  2(X_{0}^{2}-X_{2}^{2})+(X_{1}^{2}-X_{3}^{2}) \leq 0, \\
        & 2(X_{0}^{2}-X_{1}^{2}) + (X_{2}^{2}-X_{3}^{2}) \leq 0.
    \end{aligned}
    \right.
\end{equation}
It is shown in Appendix \ref{app:NEcon} that if $|\psi\rangle$ is a Nash state, the equations Eq. \eqref{eq:Nashmax4qpd} are both necessary and sufficient for $|\psi\rangle$ to be a Nash maximum state of the Quantum Prisoner's Dilemma. Imposing these inequality constraints on the Nash variety yields the blue points in Fig. \ref{fig:nvqpd}. As in Section \ref{subsec:visualizing}, the object being plotted is a double cover of the corresponding set of Nash states, which is obtained by the identification $(X_0,X_1,X_2,X_3) \sim (-X_0,-X_1,-X_2,-X_3)$. Notice that both Eqs. \eqref{eq:NV_qpd} and \eqref{eq:Nashmax4qpd} are invariant under the permutation symmetry $X_1 \leftrightarrow X_2$ interchanging the two players, which yields a reflection symmetry in the plane $x=y$ in Fig. \ref{fig:nvqpd} and is a consequence of the symmetry of the underlying classical game. 

\begin{figure}
    \centering
    \includegraphics[width=0.5\textwidth]{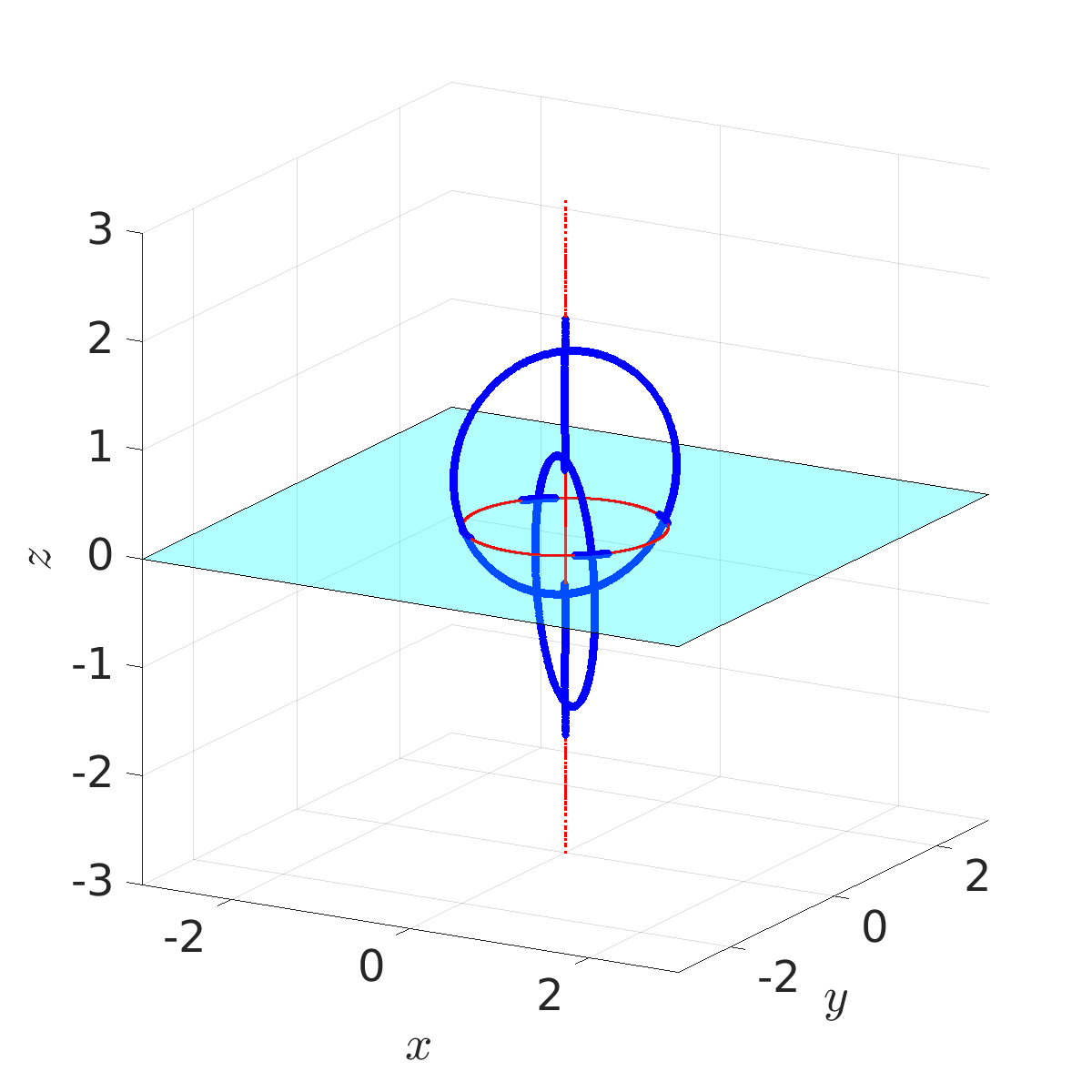}
    \caption{A visualization of the Nash variety and the set of Nash maximum states for the Quantum Prisoner's Dilemma in stereographically projected coordinates. Nash states are plotted in red. Nash states that are also Nash maximum states and therefore define pure-strategy Nash equilibria of the Quantum Prisoner's Dilemma by Lemma \ref{lemma:NashEq} are plotted in blue. The latter are obtained by imposing the inequality constraints Eq. \eqref{eq:Nashmax4qpd}.}
    \label{fig:nvqpd}
\end{figure}

We next flesh out the identification between such Nash maximum states and the Nash equilibria of Quantum Prisoner's Dilemma instances $\mathcal{G}_\mathrm{QPD}$ with specified initial states $\ket{\psi_0}$. For some intuition, note that in two-qubit Hilbert space, single-qubit rotations $G_1 \times G_2 \cong SU(2) \times SU(2)$ acting on a fixed initial state $|\psi_0\rangle$ generate an orbit in state space with constant bipartite entanglement entropy. The existence of a pure-strategy Nash equilibrium for specific instances of $\mathcal{G}_{\mathrm{QPD}}$ thus requires a non-empty intersection between this constant-entanglement orbit and the set of Nash maximum states.

When restricting to states $|\psi_0\rangle$ with real coefficients, these orbits can be characterized by quartic equations in the projected coordinates (see Appendix \ref{app:orbits} for details). Fig. \ref{fig:intersections} shows how different orbits intersect with the Nash variety. For the separable orbit Fig. \ref{fig:subfig1}, there are only two intersections with the set of Nash maximum states at $(x,y,z)= (0,0,\pm 1)$. Upon complex projectivization, both these intersection points map to the same ket $\ket{\psi^*} = \ket{11}$, which recovers the standard Nash equilibrium $(\textsc{Defect},\textsc{Defect})$ of the classical Prisoner's Dilemma. If we instead start from a maximally entangled state, the intersection between the maximally entangled orbits and the Nash variety contains four distinct points $(x,y,z) = (\pm 1, \pm 1, 0)$, as shown in Fig. \ref{fig:subfig3}. Up to an overall phase, these correspond to the states
\begin{equation}
    \ket{\psi^*} = \frac{1}{\sqrt{2}} \left( \ket{01} \pm \ket{10} \right),
\end{equation}
providing a ``quantum escape'' from the Prisoner's Dilemma with $\braket{\hat{h}_1}=\braket{\hat{h}_2}=5/2$. More precisely, quantumness of the protocol allows for \textit{correlated} Nash equilibria in which the prisoners effectively randomly choose between the strategies $(\textsc{Cooperate},\textsc{Defect})$ and $(\textsc{Defect},\textsc{Cooperate})$. 

\begin{figure*}[hbt!]
  \centering
  \subfigure[]{
    \includegraphics[width=0.3\textwidth]{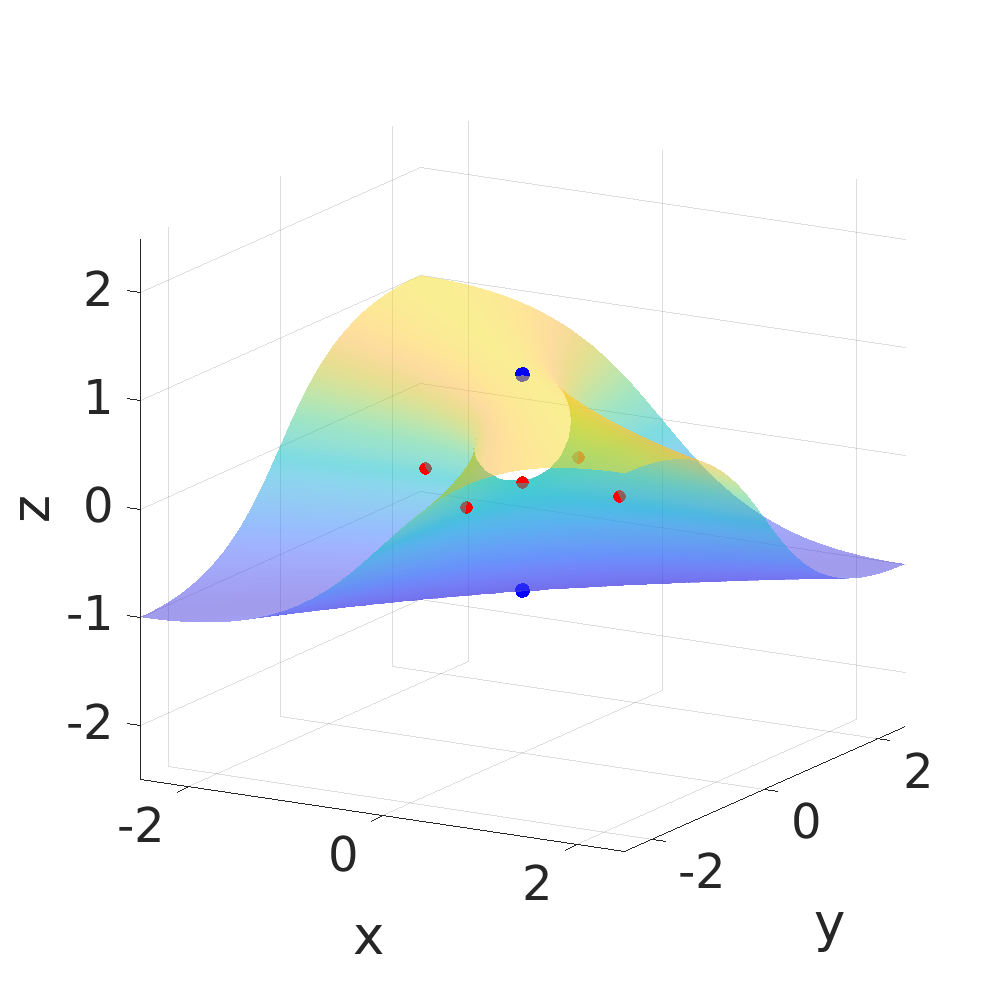}
    \label{fig:subfig1}
  }
  \subfigure[]{
    \includegraphics[width=0.3\textwidth]{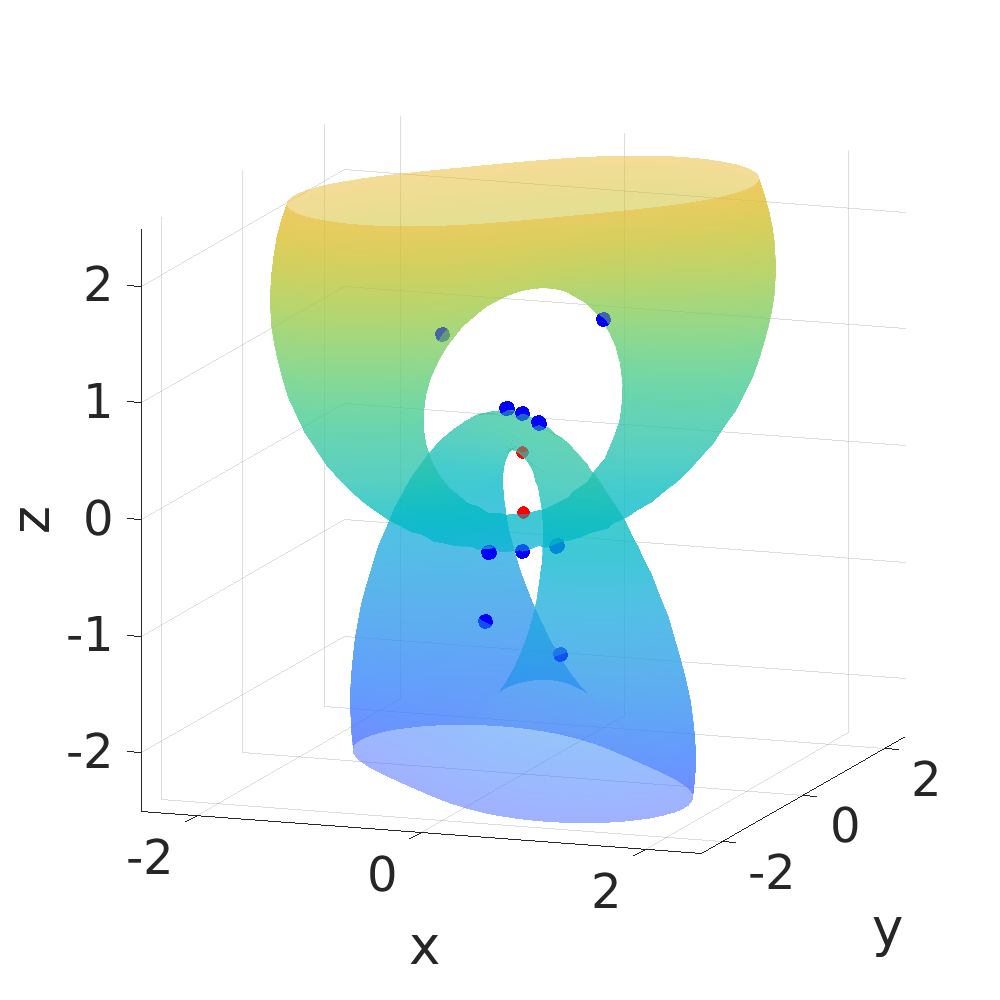}
    \label{fig:subfig2}
  }
  \subfigure[]{
    \includegraphics[width=0.3\textwidth]{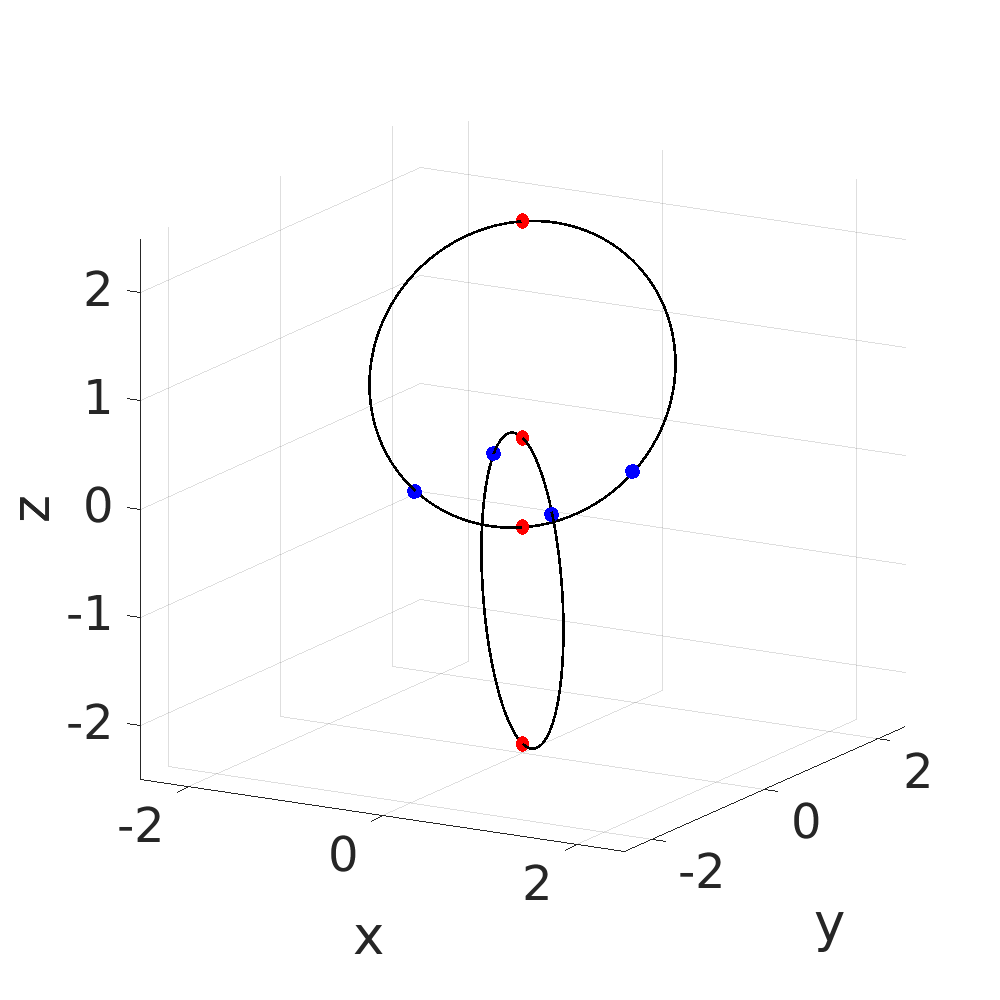}
    \label{fig:subfig3}
  }
  \caption{Visualization of the intersections of the set of Nash maximum states for the Quantum Prisoner's Dilemma, as in Fig. \ref{fig:nvqpd}, with (a) the separable orbit; (b) orbits with intermediate entanglement with $\chi = 0.22$ ; (c) maximally entangled orbits (linked circles in black). The red points indicate Nash states, while blue points are Nash maximum states on the corresponding orbits.}
  \label{fig:intersections}
\end{figure*}

The above discussion reformulates the problem of finding Nash equilibria of quantum games as a purely geometrical problem of finding intersections between orbits of states under Lie group actions and the set of Nash maximum states. By exploiting the low dimensionality of the two-rebit Hilbert space, we have further obtained and visualized the full set of Nash maximum states for the Quantum Prisoner's Dilemma. As a final remark, although the ``generic'' case depicted in Fig. \ref{fig:subfig2} is seen to host a pure-strategy Nash equilibrium, given Theorem 2 of Ref. \onlinecite{meyer_quantum_1999}, we do not expect that \emph{all} choices of initial states $|\psi_0\rangle$ will exhibit pure-strategy Nash equilibria for generic multi-observable game instances $\mathcal{G}$.

\section{Approximate Nash states}
In this section, we consider what happens when the exact pure Nash state condition Eq. \eqref{eq:Nashstate} is relaxed to an approximate extremality condition, namely the right-hand side of Eq. \eqref{eq:Nashstate} being polynomially or exponentially small in $N$ instead of exactly zero. This is relevant both for making connections with approximate extrema of variational quantum algorithms~\cite{McClean_2018,chen2023local} and for formulating standard notions of computational complexity for Nash states (for the reason that constructing exact vectors with generic complex-valued entries lies beyond any digital computation scheme). 

We first make the observation that approximate Nash states are ``ubiquitous'' in Hilbert space. Intuitively, this is because states drawn uniformly randomly from state space are close to being at infinite temperature, and the infinite temperature state is a Nash state. 

We then solidify the specific connection between Nash states and variational quantum algorithms for finding the ground states of certain classes of local Hamiltonians. Finally, we exploit this connection to the local Hamiltonian problem to make some conjectural statements on the computational complexity of approximating pure Nash equilibria of quantum games, a problem that appears to be comparable in hardness to its classical counterpart~\cite{fabrikant2004complexity}.

\subsection{Approximate Nash states are ubiquitous}
\label{sec:ubiquitous}
Equation \eqref{eq:puredimcount} implies that for large numbers of qubits $N$, the dimension of the space of pure Nash states is comparable to the exponentially large dimension of the full set of pure states. This raises the question of how pure Nash states are distributed in the space of states: are they restricted to an isolated region of state space or is there a sense in which every pure state is ``close'' to a Nash state? In this section, we show that the results of Ref. \onlinecite{chen2023local} favour the latter scenario: for a physically reasonable class of operators and Lie groups, approximate Nash states are ubiquitous in the space of states.

To this end, we first recall Lemma C.1 of Ref. \onlinecite{chen2023local}. Informally speaking, this Lemma states that given a local Hamiltonian $\hat{H}$ and a set of local unitary operations, a random pure state is with high probability an approximate local minimum of $\hat{H}$ with respect to this set of local unitary operations. Throughout this section, all operators will be understood to act on the the same $N$-qubit Hilbert space.
    
\begin{lemma}
\label{lemma:caltech}
Consider a $k$-local Hamiltonian $\hat{H}$ with $\| \hat{H} \|_{\infty} = \mathrm{poly}(N)$ and $P$ $k$-local anti-Hermitian operators $\{\hat{A}_a\}_{a=1}^P$ with $\| \hat{A}_{a} \|_{\infty}=1$, where $k=\mathcal{O}(N^0)$ and $P = \mathrm{poly}(N)$. For $N$ sufficiently large and an exponentially small error threshold $\epsilon \geq 1 / 2^{N/4}$,  a uniformly random pure state $\ket{\psi}$ is an $\epsilon$-approximate local minimum of $\hat{H}$ under unitary perturbations generated by $\{ \hat{A}_{a} \}$ with probability at least $1-1 / 2^{2^{N/4}}$.
\end{lemma}
Here ``$\epsilon$-approximate local minimum'' means that the expectation value $\bra{\psi}\hat{H}\ket{\psi}$ is $\epsilon$-stable with respect to local unitary perturbations $\ket{\psi} \mapsto \ket{\psi'}:=\exp{(\sum_{a=1}^P v_a\hat{A}_a)}\ket{\psi}$. More precisely, there exists $\delta>0$ such that for any $\| \pmb{v} \|_1<\delta$, we have
\begin{equation}
    \bra{\psi'} \hat{H}\ket{\psi'} \geq \bra{\psi} \hat{H}\ket{\psi} - \epsilon \| \pmb{v} \|_1.
\end{equation}
A key step in proving~\cite{chen2023local} Lemma \ref{lemma:caltech} is showing that for uniformly random pure states $|\psi\rangle$,
\begin{equation}
|\braket{\psi|[\hat{H},\sum_{a=1}^P v_a\hat{A}_a]|\psi}| \leq \epsilon \|\pmb{v}\|_1
\end{equation}
with high probability. This rather naturally motivates a notion of ``approximate Nash states'', which can be formalized as follows.

\begin{defn}
    A pure state $\ket{\psi}$ is an \emph{$\epsilon$-approximate Nash state} with respect to a set of observables $\{\hat{h}_i\}_{i=1}^M$ and corresponding pairwise commuting Lie groups $\{G_i\}_{i=1}^M$ of unitary operators if for all $i=1,2,\dots,M$,
    \begin{equation}
    \label{eq:approxNashstate}
    |\langle \psi | [\hat{h}_i,\pmb{v}_{i}\cdot \hat{\mathbf{A}}_{i}] |\psi\rangle| \leq \epsilon \|\pmb{v}_i\|_1
    \end{equation}
    for all $\pmb{v}_{i} \in \mathbb{R}^{\dim{(\mathfrak{g}_{i})}}$, where the vector $\hat{\mathbf{A}}_{i}=(\hat{A}_{i1},\hat{A}_{i2},\ldots,\hat{A}_{i\dim(\mathfrak{g}_i)})$ denotes a basis for $\mathfrak{g}_{i}$ with the normalization $\|\hat{A}_{i\alpha}\|_{\infty} = 1$.
\end{defn}

The following Corollary of Lemma \ref{lemma:caltech} then states that under suitable locality assumptions on the operators $\{\hat{h}_i\}_{i=1}^M$ and Lie groups $\{G_i\}_{i=1}^M$, a randomly selected pure state $\ket{\psi}$ will with high probability be an $\epsilon$-approximate Nash state with respect to these operators and Lie groups.

\begin{col}
\label{cor:ubiq}
    Let $\{ \hat{h}_{i} \}_{i=1}^{M}$ be a set of $k$-local observables with  $\| \hat{h}_{i} \|_{\infty} = \mathrm{poly}(N)$ for all $i$, and for each $i=1,2,\ldots,M$, let $G_{i}$ be a Lie group of unitary operators generated by $k$-local anti-Hermitian operators $\{\hat{A}_{i\alpha}\}_{\alpha=1}^{\dim(\mathfrak{g}_i)}$ with $\|\hat{A}_{i\alpha}\|_{\infty}=1$. Further assume that $M = \mathcal{O}(N)$ and $k=\mathcal{O}(N^0)$. Subject to these asssumptions, a uniformly random pure state $\ket{\psi}$ is a $2^{-N/4}$-approximate Nash state with probability at least $1-2^{-2^{N/4}}$.
\end{col}

The proof mirrors that of Lemma \ref{lemma:caltech} presented in Ref. \onlinecite{chen2023local} and we omit it for concision. Corollary \ref{cor:ubiq} thus yields a precise sense in which approximate Nash states are ubiquitous in the space of states. 

We note that $k$-locality of the Hamiltonians $\hat{h}_{i}$ and Lie group generators $\hat{A}_{i\alpha}$ together with the $\mathrm{poly}(N)$ bound on the number of observables are crucial for the validity of Corollary \ref{cor:ubiq}, because they guarantee that at most $\mathrm{poly}(N)$ Pauli strings appear in the conditions Eq. \eqref{eq:approxNashstate}. It would be interesting to consider how these results are altered when departing from the local setting, for example if the Lie group generators are allowed to act on an extensive number of qubits.

\subsection{Connection to variational quantum algorithms}

We now observe that our notion of Nash states of local Hamiltonians has a natural realization in the setting of low-depth quantum circuit approximations to eigenstates of local Hamiltonians. This is already apparent from the connection between Nash minimum states, optimal product states~\cite{Anshu_2021}, and local minima under local unitary perturbations~\cite{chen2023local} that we established in Section \ref{sec:solv_spec_cases} for strictly $k$-local Hamiltonians.

To develop this analogy further, consider the problem of approximating ground state energies of strictly $k$-local Hamiltonians on $N$ qubits. As noted above, this problem is rich enough to encode both $\textsc{NP}$-complete~\cite{karp2010reducibility} and $\textsc{QMA}$-complete~\cite{kempe2005complexity,piddock2015complexity} problems. Recall that a widely studied class of algorithms for constructing approximate ground states of such Hamiltonians consists of applying low-depth quantum circuits made up of one- and two-qubit unitary gates to an easy-to-prepare reference state $|\psi_0\rangle$. Perhaps the most prominent algorithm in this family is the quantum approximate optimization algorithm (QAOA)~\cite{farhi2014quantum}. From the viewpoint of this paper, any such variational quantum algorithm that seeks to minimize the expected energy $E(|\psi\rangle) = \langle \psi | \hat{H} | \psi \rangle$ with access to arbitrary single-qubit gates acting on $|\psi\rangle$ will generically converge to a Nash state in the sense of Definition \ref{def:Nash_state} rather than an eigenstate. (To see this, suppose that the final layer of the circuit is optimized over $N$ single-qubit gates and use Theorem \ref{thm:eigenstatethm}.) Since practical implementations of such algorithms involve finite times and a limited set of gates, they will furthermore realize Nash states approximately at best.

These simple observations lead to an appealing geometrical viewpoint on variational quantum algorithms (VQAs) for finding the ground state energy of strictly $k$-local Hamiltonians (the only case of such algorithms that we consider in this discussion). In the language of Section \ref{sec:geom_nash_variety}, such algorithms seek to approach a projectivized Nash variety $V'$ as closely as possible while simultaneously minimizing the energy $E(|\psi\rangle)$. We expect that the quantum computational complexity of approximately preparing states $|\psi\rangle \in V'$ will vary dramatically as this set is traversed, for the usual reason that quantum complexity is not smooth with respect to the standard inner-product metric on Hilbert space~\cite{brown2019complexity,susskind2020three,bulchandani2021smooth}. 

We emphasize that the worst-case computational hardness of approximating the energies of Nash states in no way contradicts the ubiquity of approximate Nash states discussed in the previous section. For the specific class of VQAs under consideration here, $V'$ coincides with the ``large barren plateau'' that was previously identified in the variational setting~\cite{chen2023local,McClean_2018} and argued to be ``easy to find'', on the grounds that the energy $E(|\psi\rangle)$ of a typical state $|\psi\rangle \in V'$ is computationally easy to approximate~\cite{chen2023local} by merely computing the trace of $\hat{H}$. In contrast, the computationally hard tasks that motivate studying these algorithms in the first place require showing that $E(|\psi\rangle) < E_0$ for a given threshold $E_0$ and for specific states $|\psi\rangle \in V'$, that will be very far from typical for the cases of computational interest. In more physical terms, typical states in $V'$ are effectively near infinite temperature while computationally useful states in $V'$ are effectively near zero temperature, and the computational easiness of approximating $V'$ merely reflects the overwhelming multiplicity~\cite{chen2023local} of states near infinite temperature.

To summarize, the above discussion reveals that VQAs for certain computationally interesting classes of Hamiltonians naturally realize the notions of pure Nash states and Nash varieties identified in this paper. It follows by Lemma \ref{lemma:NashEq} that such VQAs can be interpreted as solving for the pure-strategy Nash equilibria of a multi-player, non-cooperative quantum game. 

\subsection{Computational complexity of approximate Nash equilibria}

Finally, we turn to the computational complexity of approximating Nash equilibria of quantum games in light of Lemma \ref{lemma:NashEq}. There are results~\cite{bostanci_quantum_2022} indicating that approximating mixed-strategy Nash equilibria of quantum games is as computationally hard as the analogous problem for classical games~\cite{roughgarden2010algorithmic}. Our construction plausibly suggests a similar statement for pure-strategy Nash equilibria.

In more detail, consider a Heisenberg antiferromagnet with arbitrary (albeit at most polynomially large in $N$) two-body couplings. This is a strictly 2-local Hamiltonian and by Theorem \ref{thm:eigenstatethm} and Lemma \ref{lemma:NashEq}, ground states of this Hamiltonian define pure-strategy Nash equilibria of a corresponding multi-observable game. However, the problem of approximating the ground state energy of any such Hamiltonian to inverse-polynomial accuracy in $N$ is known to be complete for the quantum computational-complexity class QMA~\cite{piddock2015complexity}. This implies that approximating the payoff functions for all pure-strategy Nash equilibria of an arbitrary multi-observable game instance to inverse-polynomial accuracy is at least as hard as solving any problem in QMA. It seems reasonable to conjecture that if we demand inverse-exponential accuracy for the payoff functions, this problem becomes PSPACE-complete rather than QMA-complete~\cite{fefferman2016complete}, but solidifying this claim would require (at least) an inverse-exponential-precision generalization of the results of Ref.~\onlinecite{piddock2015complexity}. Such a result would nevertheless be satisfyingly consistent with the comparable hardness of finding quantum~\cite{bostanci_quantum_2022} versus classical~\cite{roughgarden2010algorithmic} mixed-strategy Nash equilibria, and with the PSPACE-completeness of finding pure-strategy Nash equilibria of certain classical games~\cite{fabrikant2004complexity}.

\section{Discussion}

A current trend in physics might be summarized as ``ask not what you can do for quantum mechanics, but what quantum mechanics can do for you''. As many-body physicists, we are moreover specifically interested in what quantum mechanics can do for us in the setting of scalable systems. This paper originated from previous work by the present authors and F.J. Burnell, as well as several others, examining cases in which there exists a quantum advantage for playing scalable cooperative games~\cite{bulchandani_multiplayer_2023,bulchandani_playing_2023,lin_quantum_2023,daniel_quantum_2021,hart_playing_2024}. That work motivated us to explore non-cooperative quantum games and their associated Nash equilibria in the simplest possible setting, namely static games in which each player independently selects a single move. As described in the introduction, the resulting Nash equilibria---classical or quantum---involve an optimization procedure that is very different from the optimization usually performed by physicists, raising the question of whether this idea might be useful in a physics context that is not already set up as a game.

That broad question aside, there are several directions stemming from our results in this paper that seem worth investigating further. These include exploring the various Nash varieties that can arise from different numbers and sets of observables, their geometry and topology, and the identification of local versus global Nash minima within them. The latter are mostly questions in the mathematics of infinite precision. Turning them into questions in quantum computer science and eventually ideas realizable on quantum devices requires more detailed thinking about approximate solution concepts for quantum games and measurement schemes for realizing them. While we have taken some initial steps in these directions, much remains to be done. Given the surprising connections between quantum games, algebraic varieties and the local Hamiltonian problem uncovered in this work, we are hopeful that any progress on this front will enjoy similarly wide-ranging applications.

\section{Acknowledgments}
We would like to thank F.J. Burnell for collaborations on related topics and for an early discussion on Nash equilibria of quantum games. V.B.B. thanks A. Deshpande, D.R. Foord, T.J. Osborne and B. Fefferman for helpful discussions and the Simons Institute for the Theory of Computing for their hospitality during part of the completion of this work.
S.L.S. was supported by a Leverhulme Trust International Professorship, Grant Number LIP-202-014. For the purpose of Open Access, the authors have applied a CC BY public copyright license to any Author Accepted Manuscript version arising from this submission. 
\bibliography{games}
\bibliographystyle{unsrt}

\appendix
\onecolumngrid
\section{Computation of Hessians for the TFIM}
\label{app:hess}
Consider the transverse-field Ising model with periodic boundary conditions for the spins,
\begin{equation} \label{eq:tfim}
    \hat{H}_\mathrm{TFIM} =- \sum_{i=1}^{N} \hat{Z}_{i}\hat{Z}_{i+1} - g \sum_{i=1}^{N} \hat{X}_{i}.
\end{equation}
This Appendix summarizes the computation of the expectation values $\langle x\rangle_\beta:=\braket{\hat{X}_i}_\beta$ and $\langle z z \rangle_\beta:=\braket{\hat{Z}_i \hat{Z}_{i+1}}_\beta$ at inverse temperature $\beta$ that appear in the Hessian Eq. \eqref{eq:isinghessian}. Via the standard~\cite{mbeng_quantum_2020} Jordan-Wigner transformation, one can map Eq. \eqref{eq:tfim} to a model of spinless free fermions
\begin{equation}
    \hat{H}_\mathrm{TFIM} = \hat{H}_+ \hat{P}_+ + \hat{H}_- \hat{P}_-,
\end{equation}
where $\hat{P}_\pm = \frac{1\pm (-1)^{\hat{N}}}{2}$ is the parity operator with $\hat{N}=\sum_j \hat{c}_j^\dagger \hat{c}_j$, and
\begin{equation}
    \hat{H}_{\pm} = -\sum_{j=1}^{N} (\hat{c}_{j}^{\dagger}\hat{c}_{j+1} + \hat{c}_{j}^{\dagger}\hat{c}_{j+1}^{\dagger}+\mathrm{h.c.}) - g\sum_{j=1}^{N}(1-2\hat{c}_{j}^{\dagger}\hat{c}_{j}), \quad \hat{c}_{L+1} = \mp \hat{c}_1.
\end{equation}
Diagonalizing these fermionic Hamiltonians requires Fourier transforming to Bogoliubov operators $\hat{\gamma}_k$, where $k \in K_{\sigma}$ labels the pseudo-momentum of each mode and the sets $K_\sigma \subset [0,\pi]$ depend on the fermion parity sector $\sigma=\pm$, see Table \ref{tab:modemomenta}.

\begin{table}[b]
\centering
\begin{tabular}{lll}
\toprule
& \textbf{$N$ even} & \textbf{$N$ odd} \\
\midrule
\textbf{Even parity subspace} & $K_+=\{\pm \frac{\pi}{N}, \pm \frac{3\pi}{N}, ..., \pm \frac{(N-1)\pi}{N}\}$ & $K_+=\{\pm \frac{\pi}{N}, \pm \frac{3\pi}{N}, ..., \pm \frac{(N-2)\pi}{N}, \pi\}$ \\
\textbf{Odd parity subspace} & $K_-=\{0, \pm \frac{2\pi}{N}, ..., \pm \frac{(N-2)\pi}{N}, \pi\}$ & $K_-=\{0, \pm \frac{2\pi}{N}, \pm \frac{4\pi}{N}, ..., \pm \frac{(N-1)\pi}{N}\}$ \\
\bottomrule
\end{tabular}
\caption{The allowed momenta for even/odd number of spins, in the fermionic even/odd parity sectors.}
\label{tab:modemomenta}
\end{table}

We can write the thermal expectation value of a generic observable $\hat{O}$ at inverse temperature $\beta$ as
\begin{equation}
    \begin{aligned}
        \mathrm{Tr}(\hat{O}e^{ -\beta \hat{H}_{\mathrm{TFIM}} }) & = \sum_{\sigma=\pm} \mathrm{Tr}(\hat{O}\hat{P}_{\sigma}e^{ -\beta \hat{H}_{\sigma} }) \\
        & = \frac{1}{2} \sum_{\sigma = \pm} \left( \mathrm{Tr}(\hat{O}e^{ -\beta \hat{H}_{\sigma} })+\eta_{\sigma} \mathrm{Tr}\left( \hat{O} e^{ -\beta \hat{H}_{\sigma } + i\pi \sum_{k\in K_{\sigma}} \hat{\gamma}_{k}^{\dagger}\hat{\gamma}_{k}} \right) \right),
    \end{aligned}
\end{equation}
where in the second line we defined $\eta_\sigma := \sigma \bra{0_{\sigma}}(-1)^{\hat{N}}\ket{0_{\sigma}}$, where $\ket{0_\sigma}$ is the Bogoliubov vacuum state of the Hamiltonian $\hat{H}_\sigma$. In the ferromagnetic phase, the Bogoliubov vacuum states for the Hamiltonians $\hat{H}_+,\,\hat{H}_-$ have even/odd fermionic parity respectively, whereas in the paramagnetic phase, the Bogoliubov vacuum states for $\hat{H}_+$ and $\hat{H}_-$ always have even fermionic parity. Therefore
\begin{equation}
    \eta_\sigma = \left\{
    \begin{array}{cc}
        1 & \text{when }g<1, \\
        \sigma & \text{when }g>1.
    \end{array}
    \right.
\end{equation}
It follows that the partition function
\begin{equation}
    \label{eq:partitionfunction}
    Z(\beta) = \mathrm{Tr}(e^{ -\beta \hat{H}_{\mathrm{TFIM}} }) = \frac{1}{2} \sum_{\sigma} e^{ \beta \sum_{k\in K_{\sigma}} \epsilon_{k} } \left( \prod_{k\in K_{\sigma}} (1+e^{ -2\beta \epsilon_{k} } ) + \eta_{\sigma} \prod_{k\in K_{\sigma}} (1-e^{ -2\beta \epsilon_{k} }) \right),
\end{equation}
where $\epsilon_k=\sqrt{1+g^2-2g\cos k}$ is the energy of the mode $\hat{\gamma}_k$. Letting $k\in K_\sigma$, we can then calculate the single-particle thermal averages as
\begin{equation}
\label{eq:thermalaveraged2pt}
    \begin{aligned}
        & \langle \hat{\gamma}_{k}^{\dagger} \hat{\gamma}_{k} \rangle =  \frac{1}{Z} \cdot \frac{e^{ -2\beta \epsilon_{k} }}{2} \left( \prod_{q\in K_{\sigma}, q\neq k} (1+e^{ -2\beta \epsilon_{q} }) - \eta_{\sigma} \prod_{q\in K_{\sigma, q\neq k} } (1-e^{ -2\beta \epsilon_{q} }) \right), \\
        & \langle\hat{\gamma}_{k} \hat{\gamma}_{k}^{\dagger}  \rangle = \frac{1}{Z} \cdot \frac{1}{2} \left(\prod_{q\in K_{\sigma}, q\neq k} (1+e^{ -2\beta \epsilon_{q} }) + \eta_{\sigma} \prod_{q\in K_{\sigma, q\neq k} } (1-e^{ -2\beta \epsilon_{q} })\right).
    \end{aligned}
\end{equation}
All other fermionic two-point functions (mismatched mode momenta, different parity sectors, or pair annihilation/creation operators) vanish. From Eq. \eqref{eq:thermalaveraged2pt}, we can thus obtain the thermal averages $\langle x \rangle_\beta$, $\langle z z \rangle_\beta$ as
\begin{equation}
    \begin{aligned}
         \langle x \rangle_{\beta}  =& -1+\frac{2}{N} \sum_{\sigma} \sum_{k\in K_{\sigma}} \left( \cos^{2} \frac{\theta_{k}}{2} \langle \hat{\gamma}_{k} \hat{\gamma}_{k}^{\dagger}\rangle + \sin^{2} \frac{\theta_{k}}{2} \langle \hat{\gamma}_{k}^{\dagger} \hat{\gamma}_{k} \rangle \right), \\
        \langle z z \rangle_{\beta} =& -\frac{1}{N} \sum_{\sigma}\sum_{k\in K_{\sigma}} e^{ -ik }\cos^{2}\frac{\theta_{k}}{2}   \langle \hat{\gamma}_{k} \hat{\gamma}_{k}^{\dagger}\rangle + e^{ik}\sin^{2} \frac{\theta_{k}}{2} \langle\hat{\gamma}_{k}^{\dagger} \hat{\gamma}_{k}\rangle + \mathrm{h.c.} \\
        &-\frac{1}{N} \sum_{\sigma} \sum_{k\in K_{\sigma}} -ie^{ -ik } \sin \frac{\theta_{k}}{2} \cos \frac{\theta_{k}}{2} \langle \hat{\gamma}_{k} \hat{\gamma}_{k}^{\dagger}\rangle-ie^{ ik } \sin \frac{\theta_{k}}{2} \cos \frac{\theta_{k}}{2} \langle\hat{\gamma}_{k}^{\dagger} \hat{\gamma}_{k}\rangle + \mathrm{h.c.}
    \end{aligned}
\end{equation}
where $\theta_k$ are the ``Bogoliubov angles'' for diagonalizing the fermionic Hamiltonians, given by~\cite{mbeng_quantum_2020}
\begin{equation}
    \sin \theta_k = \frac{\sin k}{\epsilon_k}, \quad \cos \theta_k  = \frac{g-\cos k}{\epsilon_k}.
\end{equation}

\section{TFIM ground states are Nash minimum states}\label{app:GSareNE}

In this Appendix, we prove the claim in Section \ref{sec:temperature} that ground states of the TFIM $\hat{H} = \sum_{i=1}^N \hat{h}_i$ with $\hat{h}_i$ as in Eq. \eqref{eq:local_decomp} are Nash minimum states of $\{\hat{h}_i\}_{i=1}^{N}$ with respect to single-qubit rotations $G_i \cong SU(2)$. Using the notation of Appendix \ref{app:hess}, we can write this ground state explicitly as 
\begin{equation}
|\mathrm{GS}\rangle=\prod_{0<k<\pi}\left(\cos \left(\frac{\theta_k}{2}\right)-\sin \left(\frac{\theta_k}{2}\right) \hat{c}_k^{\dagger} \hat{c}_{-k}^{\dagger}\right)|0\rangle.
\end{equation}
We will parameterize the unitary transformations $\hat{U}_{i} \in G_{i}$ as
\begin{equation}
\hat{U}_{i} =  \cos{\left(\frac{\phi_i}{2}\right)}\hat{\mathbbm{1}} + \mathrm{i} \sin{\left(\frac{\phi_i}{2}\right)}\mathbf{n}_{i}\cdot \hat{\pmb{\sigma}}_i,
\label{eq:SU2parameterization}
\end{equation}
where $\mathbf{n}_{i}=(n_{i}^{x}, n_{i}^{y}, n_{i}^{z})$ are unit vectors in $\mathbb{R}^{3}$, $\phi_{i} \in [0,\pi]$ and $\hat{\pmb{\sigma}}_i = (\hat{\sigma}^x_i,\hat{\sigma}^y_i,\hat{\sigma}^z_i)$. It suffices to show that
\begin{equation}
\bra{\mathrm{GS}} \hat{h}_{i} \ket{\mathrm{GS}} \leq \bra{\mathrm{GS}} \hat{U}_{i}^{\dagger} \hat{h}_{i} \hat{U}_{i}\ket{\mathrm{GS}}
\end{equation}
for all $\hat{U}_i \in G_i$. Expanding with the parameterization in Eq. \eqref{eq:SU2parameterization}, we have
\begin{equation}
\begin{aligned}
     \bra{\mathrm{GS}} \hat{U}_{i}^{\dagger} \hat{h}_{i} \hat{U}_{i}\ket{\mathrm{GS}}
 = &\cos^2{\left(\frac{\phi_i}{2}\right)} \bra{\mathrm{GS}} \hat{h}_i \ket{\mathrm{GS}} + \mathrm{i} \cos {\left(\frac{\phi_i}{2}\right)} \sin {\left(\frac{\phi_i}{2}\right)}  \bra{\mathrm{GS}} [\hat{h}_i, \mathbf{n}_i\cdot \hat{\pmb{\sigma}}_i] \ket{\mathrm{GS}} \\ &\quad + \sin^2{\left(\frac{\phi_i}{2}\right)} \bra{\mathrm{GS}} (\mathbf{n}_i\cdot \hat{\pmb{\sigma}}_i) \hat{h}_i(\mathbf{n}_i\cdot \hat{\pmb{\sigma}}_i) \ket{\mathrm{GS}}.
\end{aligned}
\label{eq:IsingPayoffExpansion}
\end{equation}
In the second term, the commutators can be computed as
\begin{equation}
\begin{aligned}
        &[\hat{h}_i,\hat{X}_i] = -\mathrm{i} \hat{Y}_i \left(\hat{Z}_{i-1}+\hat{Z}_{i+1}\right),  \\
        &[\hat{h}_i,\hat{Y}_i] = -\mathrm{i}2g\hat Z_i+\mathrm{i}\hat{X}_i\left(\hat{Z}_{i-1}+\hat{Z}_{i+1}\right), \\
        &[\hat{h}_i,\hat{Z}_i] = \mathrm{i}2g\hat Y_i.
\end{aligned}\label{eq:singlequbitcommutators}
\end{equation}
All these commutators create or destroy unpaired Jordan-Wigner fermions and their ground-state expectation values therefore vanish. For the last term in Eq. \eqref{eq:IsingPayoffExpansion}, we expand $\mathbf{n}\cdot \hat{\pmb{\sigma}}$ and eliminate terms that change the overall fermionic parity. This yields
\begin{equation}
\begin{aligned}
\bra{\mathrm{GS}} (\mathbf{n}_i\cdot \hat{\pmb{\sigma}}_i) \hat h_i (\mathbf{n}_i\cdot \hat{\pmb{\sigma}}_i) \ket{\mathrm{GS}} &= (n_i^x)^2\braket{\hat X_i \hat h_i \hat X_i} + (n_i^y)^2\braket{\hat Y_i \hat h_i \hat Y_i} + (n_i^z)^2\braket{\hat Z_i \hat h_i \hat Z_i} + n_i^yn_i^z \braket{\hat Y_i \hat h_i \hat Z_i+\hat Z_i \hat h_i \hat Y_i} \\
& = \langle \hat{Z}_i \hat{Z}_{i+1} \rangle \left((n_i^x)^2+(n_i^y)^2-(n_i^z)^2\right) + g \langle \hat{X}_i\rangle \left(-(n_i^x)^2+(n_i^y)^2+(n_i^z)^2\right) \\ & \geq -\langle  \hat Z_i \hat Z_{i+1} + g\hat X_i \rangle.
\end{aligned}\label{eq:NonVanishingExpansionTerms}
\end{equation}
where we used translation invariance and fermion parity in going to the second line and non-negativity of the ground-state expectation values $\langle \hat{X}_{i} \rangle$ and $\langle \hat{Z}_{i} \hat{Z}_{i+1} \rangle$ for $g>0$ and the unit norm of $\mathbf{n}_i$ in going to the last line. Substituting Eq. \eqref{eq:NonVanishingExpansionTerms} into Eq. \eqref{eq:IsingPayoffExpansion} gives
\begin{equation}
 \begin{aligned}
\bra{\mathrm{GS}} U_{i}^{\dagger} \hat{h}_{i} \hat{U}_{i}\ket{\mathrm{GS}} &\geq - \cos^2{\left(\frac{\phi_i}{2}\right)} \bra{\mathrm{GS}} \hat Z_i \hat Z_{i+1} + g\hat X_i \ket{\mathrm{GS}} - \sin^2{\left(\frac{\phi_i}{2}\right)} \bra{\mathrm{GS}} \hat Z_i \hat Z_{i+1} + g\hat X_i \ket{\mathrm{GS}} \\
& = \bra{\mathrm{GS}} \hat{h}_{i} \ket{\mathrm{GS}}. 
\end{aligned}
\end{equation}
Thus we have shown that the TFIM ground state $\ket{\mathrm{GS}}$ is a Nash minimum state in the sense of Definition \ref{def:Nash_eqstate}, with respect to the operators $\{\hat{h}_{i}\}_{i=1}^N$ in Eq. \eqref{eq:local_decomp} and single-qubit rotations $G_i \cong SU(2)$.

\section{Nash equilibria of the Quantum Prisoner's Dilemma}

In this Appendix, we present some technical details underpinning the results of Section \ref{sec:case_study}. 

\subsection{``Rebits'' are enough} \label{app:rebits}
We first explain why examining ``rebits'', as in Eq. \eqref{eq:rebits}, is sufficient to capture the full structure of Nash states for the Quantum Prisoner's Dilemma. Let
\begin{equation}
    R_{\mathrm{QPD}} := \{\ket{\psi_0} \in \mathbb{R}^4 \subseteq \mathbb{C}^4| \bra{\psi_0} \mathrm{i}[\hat{h}_i, \hat{\sigma}_{i\alpha}]\ket{\psi_0} = 0, \ i=1,2, \ \alpha = 1,2,3 \}
\end{equation}
denote the variety of rebit Nash state vectors.  Given a state vector $\ket{\psi_0} \in R_\mathrm{QPD}$, one can immediately generate a three-torus of Nash state vectors $|\psi\rangle \in \mathbb{C}^4$ with constant payoff functions $\langle \psi | \hat{h}_i | \psi \rangle$ by applying the product of single-qubit gates $e^{\mathrm{i}(\alpha_0 + \alpha_1 \hat{Z}_1+\alpha_2 \hat{Z}_2)}$, with $\alpha_i\in [0,2\pi)$. Such operations transform $|\psi_0\rangle$ as
\begin{equation}
    e^{\mathrm{i}(\alpha_0+\alpha_1 \hat{Z}_1+\alpha_2 \hat{Z}_2)} \ket{\psi_0} = e^{\mathrm{i}(\alpha_0+\alpha_1+\alpha_2)}\left( X_0\ket{00} + X_1 e^{-2\mathrm{i}\alpha_2}\ket{01} + X_2 e^{-2\mathrm{i}\alpha_1}\ket{10} + X_3 e^{-2\mathrm{i}(\alpha_1+\alpha_2)}\ket{11}\right).
\end{equation}
We argue that the full Nash variety for the Quantum Prisoner's Dilemma is given by the product of $R_{\mathrm{QPD}}$ with a three-torus, namely
\begin{equation}
\label{eq:fullpara}
    V_{\mathrm{QPD}} := \{\ket{\psi} \in \mathbb{C}^4 | \ket{\psi} = e^{\mathrm{i}(\alpha_0+\alpha_1 \hat{Z}_1 +\alpha_2 \hat{Z}_2)} \ket{\psi_0},\ \ket{\psi_0} \in R_\mathrm{QPD}, \ \alpha_0,\alpha_1,\alpha_2 \in [0,2\pi) \}.
\end{equation}
It can be deduced from the fact that $[\hat{h}_i,\hat{Z}_i]=0$ that $V_\mathrm{QPD}$ is contained within the Nash variety. Now consider an arbitrary Nash state vector $\ket{\psi} \in \mathbb{C}^4$, which can be written as
\begin{equation}
    \ket{\psi} = X_0e^{i\phi_0}\ket{00} +  X_1 e^{i\phi_1} \ket{01} + X_2 e^{\mathrm{i}\phi_2} \ket{10} + X_3 e^{\mathrm{i} \phi_3}\ket{11},
\end{equation}
where $(X_0,X_1,X_2,X_3)\in \mathbb{R}^4$ and $\phi_0,\phi_1,\phi_2,\phi_3 \in [0,2\pi)$. Then acting with $e^{\mathrm{i}(\alpha_0+\alpha_1 \hat{Z}_1 +\alpha_2 \hat{Z}_2)}$ with $\alpha_0 = -(\phi_1+\phi_2)/2, \, \alpha_1 = (\phi_2-\phi_0)/2, \,  \alpha_2 = (\phi_1-\phi_0)/2$ yields another Nash state vector
\begin{equation}
    \ket{\psi'} = X_0\ket{00} + X_1 \ket{01} + X_2 \ket{10} + X_3 e^{\mathrm{i}\phi_3'} \ket{11}
\end{equation}
with $\phi_3' = \phi_0 + \phi_3-\phi_1 - \phi_2$. We define new real variables $X_4,\,X_5$ by $X_3e^{\mathrm{i}\phi_3'} = X_4 + \mathrm{i} X_5$ and write down the Nash state conditions explicitly:
\begin{equation} \label{eq:seteqs}
\left\{
    \begin{aligned}
        & 2 X_0 X_2 + X_1 X_4 = 0, \\
        & 2 X_0 X_1 +  X_2 X_4 = 0, \\
        & X_1 X_5 = 0, \\
        & X_2 X_5 = 0.
    \end{aligned}\right.
\end{equation}
Solutions to Eq. \eqref{eq:seteqs} satisfy either $X_5=0$ or $X_1=X_2=0$. The former case immediately yields $\ket{\psi'} \in R_{\mathrm{QPD}}$, whereas the latter case satisfies $e^{\mathrm{i}(-\phi_3'/2 + (\phi_3'/2)\hat{Z}_1)}\ket{\psi'} \in R_{\mathrm{QPD}}$. We deduce that $V_{\mathrm{QPD}}$ indeed captures the entire Nash variety. It follows that consideration of $R_{\mathrm{QPD}}$ alone is sufficient to capture all possible values of the payoff functions for Nash states of the Quantum Prisoner's Dilemma, and that all points of $V_{\mathrm{QPD}}$ are simply related to points of $R_{\mathrm{QPD}}$ via the torus action in Eq. \eqref{eq:fullpara}.

To summarize, up to assigning such tori to individual points of $R_{\mathrm{QPD}}$, the visualizations in Fig. \ref{fig:nvqpd} and Fig. \ref{fig:intersections} encapsulate the full set of Nash states for the Quantum Prisoner's Dilemma.

\subsection{The Nash equilibrium condition} \label{app:NEcon}
We now show that if $|\psi\rangle$ is a Nash state for the Quantum Prisoner's Dilemma with respect to single-qubit rotations, then it is additionally a Nash maximum state iff Eq. \eqref{eq:Nashmax4qpd} holds.  

We start by writing the payoff Hamiltonians Eq. \eqref{eq:qpdpayoffs} in terms of Pauli matrices
\begin{equation} \label{eq:qpdpayoff_rewrite}
    \begin{aligned}
\hat h_1 &= \frac{9}{4} \hat{\mathbbm{1}} \otimes \hat{\mathbbm{1}}   + \frac{7}{4} \hat{\mathbbm{1}}  \otimes \hat Z   - \frac{3}{4} \hat Z \otimes \hat{\mathbbm{1}}  -\frac{1}{4} \hat Z \otimes \hat Z,\\
    \hat h_2 &= \frac{9}{4} \hat{\mathbbm{1}}  \otimes \hat{\mathbbm{1}}  - \frac{3}{4} \hat{\mathbbm{1}}  \otimes \hat Z + \frac{7}{4} \hat Z \otimes \hat{\mathbbm{1}}  -\frac{1}{4} \hat Z\otimes \hat Z.
\end{aligned}
\end{equation}

For $\ket{\psi}$ to be a Nash maximum state, we require $\bra{\psi}U_{i}^{\dagger} \hat{h}_{i} U_{i}-\hat{h}_{i}\ket{\psi} \leq 0$ for $i=1,2$ and for all $U_{i}\in SU(2)$. Using the same parameterization of single-qubit rotations as in Eq. \eqref{eq:SU2parameterization} and the Nash state property of $|\psi\rangle$, we have
\begin{equation}
\bra{\psi}U_{i}^{\dagger} \hat{h}_{i} U_{i}-\hat{h}_{i}\ket{\psi} = \sin^2{\left(\frac{\phi_i}{2}\right)} \left[\bra{\psi} (\mathbf{n}_{i}\cdot \hat{\pmb{\sigma}}_i)\hat{h}_{i} (\mathbf{n}_{i}\cdot \hat{\pmb{\sigma}}_i) - \hat{h}_{i}\ket{\psi}\right], \quad i=1,2,
\end{equation}
and it follows that $\ket{\psi}$ is Nash maximal iff
\begin{equation}
\label{eq:NEcon_rewrite1}
\bra{\psi} (\mathbf{n}_{i}\cdot \hat{\pmb{\sigma}}_i) \hat{h}_{i} (\mathbf{n}_{i}\cdot \hat{\pmb{\sigma}}_i) \ket{\psi} \leq \bra{\psi} \hat{h}_{i}\ket{\psi}, \quad i=1,2.
\end{equation}

Elementary manipulation of Pauli matrices now yields the following simplifications:
\begin{enumerate}
\item  $\hat{X}_i\hat{h}_i\hat{X}_i=\hat{Y}_i\hat{h}_i\hat{Y}_i$ and  $\hat{Z}_i\hat{h}_i\hat{Z}_i=\hat{h}_i$,
    \item $\hat{X}_i\hat{h}_i\hat{Y}_i+\hat{Y}_i\hat{h}_i\hat{X}_i=0$,
    \item $\hat{Z}_i\hat{h}_i\hat{Y}_i+\hat{Y}_i\hat{h}_i\hat{Z}_i = -\mathrm{i}[\hat{h}_i,\hat{X}_i]$ and $\hat{X}_i\hat{h}_i\hat{Z}_i+\hat{Z}_i\hat{h}_i\hat{X}_i=\mathrm{i}[\hat{h}_i,\hat{Y}_i]$, which have vanishing expectation values in the state $|\psi\rangle$ by its Nash state property.
\end{enumerate}
Substituting the above facts into Eq. \eqref{eq:NEcon_rewrite1} yields necessary and sufficient conditions
\begin{equation}
    \label{eq:NEcons_final}
    \langle \psi | \hat{X}_{i}\hat{h}_{i}\hat{X}_{i} | \psi \rangle \leq\langle \psi | \hat{h}_{i} | \psi \rangle, \quad i=1,2,
\end{equation}
for a given Nash state $|\psi\rangle$ to be Nash maximal. Parameterizing the (rebit) state $\ket{\psi}$ explicitly as $\ket{\psi} =X_{0}\ket{00} +X_{1}\ket{01} +X_{2}\ket{10} + X_{3} \ket{11} $, Eq. \eqref{eq:NEcons_final} becomes
\begin{equation}
    \label{eq:NEcons_final_S3}
    2(X_{0}^{2}-X_{2}^{2})+(X_{1}^{2}-X_{3}^{2}) \leq 0, \quad 2(X_{0}^{2}-X_{1}^{2}) + (X_{2}^{2}-X_{3}^{2}) \leq 0,
\end{equation}
recovering Eq. \eqref{eq:Nashmax4qpd}. In terms of the stereographically projected coordinates Eq. \eqref{eq:stereoprojection}, these conditions read
\begin{equation}
    \left\{
    \begin{aligned}
        & (1-x^{2}-y^{2}-z^{2})^{2}+2x^{2}-4y^{2}-2z^{2} \leq 0, \\
        & (1-x^{2}-y^{2}-z^{2})^{2}-4x^{2}+2y^{2}-2z^{2} \leq 0,
    \end{aligned}
    \right.
\end{equation}
as depicted in Fig. \ref{fig:nvqpd}.

\subsection{Orbits of two rebits under the action of $SU(2)\times SU(2)$} \label{app:orbits}

Finally, we discuss the stratification of two-qubit Hilbert space $\mathcal{H} \cong \mathbb{C}^4$ into orbits of the $SU(2)\times SU(2)$ action induced by single-qubit rotations. As in Section \ref{sec:case_study}, we restrict attention to unit-normalized rebit states $|\psi\rangle = X_0 \ket{00} + X_1 \ket{01} + X_2 \ket{10} + X_3 \ket{11}$, which greatly simplifies the problem and allows for an explicit visualization as in Fig. \ref{fig:intersections}.  (For the full pure-state space $\cp{3}$, a comprehensive discussion is given in Ch. 16 of Ref. \onlinecite{bengtsson_geometry_2006}.)

For such states, the bipartite entanglement entropy is determined by the eigenvalues of the one-qubit reduced density matrix (RDM),
\begin{equation}
    \hat{\rho}_{1} = \begin{pmatrix}
X_{0}^{2}+X_{1}^{2}  &  X_{0}X_{2}+X_{1}X_{3} \\ 
X_{0}X_{2}+X_{1}X_{3}  &  X_{2}^{2}+X_{3}^{2}
\end{pmatrix}.
\end{equation}
Its eigenvalues are
\begin{equation}
    \lambda_{\pm} = \frac{1}{2} \pm \frac{1}{2}\sqrt{ 1-4(X_{0}X_{3}-X_{1}X_{2})^{2} }.
\end{equation}
We define the entanglement parameter $\chi^2:=(X_{0}X_{3}-X_{1}X_{2})^{2} \in \left[ 0, \frac{1}{4} \right]$. The state $\ket{\psi}$ is separable if $\chi^2=0$ and maximally entangled if $\chi^2=1/4$.
\subsubsection{The separable orbit}

From the above discussion, the criterion for separability of $\ket{\psi}$ is $X_{0}X_{3}-X_{1}X_{2}=0$. In stereographically projected coordinates, one has
\begin{equation}
    \label{eq:separableorbit}
    (1-x^{2}-y^{2}-z^{2})z-2xy = 0.
\end{equation}
As illustrated in Fig. \ref{fig:subfig1}, the stereographically projected separable orbit is homeomorphic to the complement of a pair of linked circles in $\R{3}$. It intersects the $z$-axis at $\ket{00}$ and $\ket{11}$.

\subsubsection{Maximally entangled orbits}

Maximally entangled and normalized state vectors with real coefficients form a one-dimensional subset of $S^3$. A necessary and sufficient condition for maximal entanglement is that the matrix
\begin{equation}
    \sqrt{2}\begin{pmatrix}
X_{0} & X_{1} \\
X_{2} &  X_{3} 
\end{pmatrix}
\end{equation}
is unitary, i.e. $X_{0}=X_{3},X_{1}=-X_{2}$, or $X_{0}=-X_{3},X_{1}=X_{2}$. In projected coordinates, these constraints read
\begin{equation}
    x=-y, 1-x^2-y^2-z^2 = 2z; \text{  or  } x = y, 1-x^2-y^2-z^2 = -2z,
\end{equation}
which can be written as two disjoint, linked circles:
\begin{equation}
    \left\{
\begin{aligned}
& x=-y, \, \, \,  (z-1)^{2} + x^{2} + y^{2} = 2, \\
&  x=y, \, \, \,  (z+1)^{2} + x^{2} + y^{2} = 2.
\end{aligned}
\right .
\end{equation}
As illustrated in Fig. \ref{fig:subfig3}, these orbits intersect the $z$-axis at $\ket{\mathrm{GHZ}^\pm} = \frac{1}{\sqrt{2}}(\ket{00} \pm \ket{11})$, and the circle $x^2+y^2=1,z=0$ at the other two Bell states $\frac{1}{\sqrt{2}} (\ket{01} \pm \ket{10})$.

\subsubsection{Orbits with intermediate entanglement}

 In stereographically projected coordinates, the defining equations for ``generic'' orbits with entanglement parameter $0<\chi^2<1/4$ are 
\begin{equation}
    (1-x^{2}-y^{2}-z^{2})z-2xy \pm \chi (1+x^{2}+y^{2}+z^{2})^{2}= 0.
\end{equation}
As illustrated in Fig \ref{fig:subfig2}, these orbits appear to comprise two linked tori symmetric under inversion with respect to the origin. This shape interpolates between the separable and the maximally entangled orbits. For example, the separable orbit can be seen as a limiting case in which the tori expand until their boundaries are touching one another.

\end{document}